\documentclass[11pt]{article}
\usepackage[a4paper,hmargin=3cm,vmargin=3cm]{geometry}

\usepackage{eqnarray,amssymb,amsmath,amsthm,mathrsfs,array,tikz, authblk}
\usepackage{tikz-cd}
\usepackage[applemac]{inputenc} 
\usepackage[T1]{fontenc} 
\usepackage{tipa,url}

\newcommand{\MM}{\mathcal{M}}

\newcommand{\DD}{\mathscr{D}}
\newcommand{\OO}{\mathcal{O}}

\newcommand{\valua}{\mbox{\textscriptv}}

\newcommand{\pcal}{\mathfrak{p}}
\newcommand{\FF}{{\mathcal F}}
\newcommand{\F}{{\mathbb F}}
\newcommand{\Z}{{\mathbb Z}}
\newcommand{\N}{{\mathbb N}}
\newcommand{\E}{{\mathscr E}}
\newcommand{\C}{{\mathscr C}}
\newcommand{\GF}[1]{\F_{#1}}
\newcommand{\GFbar}[1]{\overline{{\mathbb F}}_{#1}}
\newcommand{\PP}{{\mathbb P}}
\newcommand{\EE}{{\mathbb K}}
\newcommand{\LL}{{\mathbb L}}

\newcommand{\bracket}[2]{\left[{#1},{#2}\right]}
\newcommand{\Pol}{\mbox{Pol}}

\newcommand{\Fr}{\pi}
\newcommand{\Deg}{\mbox{Deg}}
\newcommand{\Princ}{\mbox{Princ}}
\newcommand{\Pic}{\mbox{Pic}}
\newcommand{\Div}{\mbox{Div}}

\newtheorem{defi}{Definition}
\newtheorem{lemma}{Lemma}
\newtheorem{remark}{Remark}
\newtheorem*{heuristic}{Heuristic}

\newtheorem{theorem}{Theorem}
\newtheorem{corollary}{Corollary}
\newtheorem{propo}{Proposition}

\DeclareMathAlphabet{\mathpzc}{OT1}{pzc}{m}{it}

\tikzset{ampersand replacement=\&}

\definecolor{marronchene}{RGB}{217,205,180}
\definecolor{vertdragee}{RGB}{180,217,190}
\definecolor{bleulavande}{RGB}{183,180,217}
\definecolor{bleuciel}{RGB}{240,240,250}
\definecolor{rougegarance}{RGB}{217,180,180}
\definecolor{grisargent}{RGB}{210,210,220}
\definecolor{oeuf}{RGB}{225,215,210}

\definecolor{vertforet}{RGB}{50,150,50}
\definecolor{bleumarin}{RGB}{50,50,150}
\definecolor{rougerubis}{RGB}{150,50,50}
\definecolor{grisouris}{RGB}{100,100,100}
\definecolor{bleucecile}{RGB}{50,100,110}
\definecolor{vertcecile}{RGB}{50,110,80}

\begin{document}

\newif\ifanonymous
\anonymousfalse 

\title{Algorithmic aspects of elliptic bases in finite field discrete logarithm algorithms}
\author[1,2]{Antoine Joux}
\author[3]{C\'ecile Pierrot}
 \affil[1]{\footnotesize Sorbonne Universit\'e, Institut de Math\'ematiques de
  Jussieu--Paris Rive Gauche, CNRS, INRIA, Univ Paris Diderot. Campus Pierre et Marie Curie, F-75005, Paris, France}
  \affil[2]{\footnotesize Chaire de Cryptologie de la Fondation SU}
  \affil[3]{\footnotesize  Université de Lorraine, CNRS, Inria, LORIA, F-54000 Nancy, France}
\date{}

\maketitle

\begin{abstract}
 Elliptic bases, introduced by Couveignes and Lercier in 2009, give an
 elegant way of representing finite field extensions. A natural
 question which seems to have been considered independently by several
 groups is to use this representation as a starting point for small
 characteristic finite field discrete logarithm algorithms.

 This idea has been recently proposed by two groups working on it, in
 order to achieve provable quasi-polynomial time for discrete
 logarithms in small characteristic finite fields.

 In this paper, we don't try to achieve a provable algorithm but,
 instead, investigate the practicality of heuristic algorithms based
 on elliptic bases. Our key idea, is to use a different model of the
 elliptic curve used for the elliptic basis that allows for a
 relatively simple adaptation of the techniques used with former
 Frobenius representation algorithms.
 
 We haven't performed any record computation with this new method but
 our experiments with the field $\F_{3^{1345}}$ indicate that
 switching to elliptic representations might be possible with
 performances comparable to the current best practical methods.
\end{abstract}

\section{Introduction}
The discrete logarithm problem (DLP) is a fundamental
problem underlying the security of many cryptographic systems. Given
$G$  a finite cyclic group denoted multiplicatively and $g$ a
generator of the group, solving the discrete logarithm problem 
in $G$ means being able, for any arbitrary element $h \in G$, to 
find an integer $x$ such that:
\[g^x = h.
\]
The integer $x$ is defined modulo the order of $G$ and is called the
discrete logarithm of $h$.

Among the groups considered for cryptographic use, we find the
multiplicative group of finite fields. There is a long history of
algorithms to address this problem that we do not recall here. In the
case of ``small'' characteristic fields, tremendous progress was made
in 2013 and the years after. They are surveyed
in~\cite{DBLP:journals/dcc/JouxP16}.  This led to extreme
computational improvements and to two flavors of heuristic
quasi-polynomial time algorithms.  One of the fundamental tools used
to achieved this result is a special representation of a finite field
extension above $\F_q$, called the Frobenius representation which
requires an element $\theta$ satisfying a relation of the form:
\[\theta^q = \frac{h_0(\theta)}{h_1(\theta)},
\]
where $h_0$ and $h_1$ are co-prime polynomials with very low degree.

A widely believed heuristic assumption is that any finite field
extension can be represented in that way, unless one of the known
obstructions  applies. These known obstructions are that $h_0$ and $h_1$
can both have degree $\leq 1$. Furthermore, it is clearly not possible
to represent an extension of degree higher than $q+\deg(h_1).$
In practice, finding such a representation via exhaustive search among
suitable polynomials is a trivial matter. However, proving this
assumption seems to be a difficult task.

As a consequence, it is natural to turn to different field
representations which can provably be constructed and try to adapt the
discrete logarithm algorithms to work with them.  Elliptic bases, also
called elliptic periods~\cite{DBLP:journals/ffa/CouveignesL09}, form a
natural candidate for this purpose. Their use for discrete logarithms
was independently considered by several groups. We are aware of two
attempts which have been made public. In 2016, in his master's
thesis~\cite{GuidoLido}, Lido proposed a discrete logarithm based on
elliptic representations, using a descent method made of two halves.
His presentation states a theorem concerning one half of the descent
and a conjecture for the other half.  On June 26th, 2019, Kleinjung
and Wesolowski released a preprint~\cite{KW19} on the eprint archive
announcing a fully provable quasi-polynomial time algorithm based on
elliptic representation.  The next day, Schoof gave a talk at the
conference NutMiC~2019 presenting the work of Lido. He also sent us a
not yet publicly available document~\cite{GuidoLido2} that
extends~\cite{GuidoLido} and contains a theorem announcing an
algorithm to compute logarithms in a finite field $\F$ in provable
time $(\log |\F|)^{O(\log\log |\F|)})$.

The result announced in~\cite{KW19}, has a different form. It states
that discrete logarithms in $\F_{p^n}$ can be computed in provable
time $(pn)^{2\log_2{n}+O(1)}.$ Both forms affect a very
large range of characteristic. Indeed, until now, the best provable
discrete logarithm algorithms for finite fields had complexity
$L(1/2)$. They are thus outperformed as soon as
$p<L_{p^n}(1/2-\epsilon)$, for an arbitrary small $\epsilon>0$.

In this paper, we present the work we independently performed on a similar
idea. However, we do not consider the provable aspects. Instead, we focus more
on the algorithmic aspects of heuristic variants of elliptic
representation discrete logarithm methods. Our formulation differs in
many details. As such, it might shed a different light on the topic
and help the reader to study the theoretical breakthrough on provable
algorithms. As of now, our proposal remains slightly inferior to the
method in~\cite{DBLP:conf/asiacrypt/JouxP14}, the fastest currently known (heuristic) method to compute discrete logarithms in small characteristic. However, it gets very
close to it, while leaving room for improvement in the analysis.

\paragraph{Quick overview of Frobenius representation algorithms.}
The common strategy of all heuristic algorithms in the Frobenius representation
family is the following.

\begin{itemize}
\item
A preliminary step, often called the  {\it representation phase}, we
construct a representation of an extension of degree $k$ of $\F_q$, assuming that:
\begin{heuristic} 
\label{heuristicFrobenius}
There exists a small divisor $n$ of $k$ (possibly one), two coprime
polynomials~$h_0$ and~$h_1$ in~$\F_{q^n}[X]$ of low degree (often of
degree at most two) and an irreducible factor~$I$ of~$h_1 X^q - h_0$
having degree~$k/n$.
\end{heuristic} This gives the representation of the target field
$\F_{q^k}$ as $\F_{q^n}[X]/(I)$. 
\item { \bf Relation generation.} Set a small set $\FF$ of particular elements, small in some sense, such that $\FF \subset \F_{q^k}$. Collect vectors $(e_f)_{f \in \FF}$ such that:
\[\prod_{f \in \FF} f^{e_f} = 1.
\]These vectors create linear relations between the discrete logarithm of factor base elements since:
\[\sum_{f \in \FF} {e_f}. \log_g f = 0.
\]
\item { \bf Linear algebra.} Find a solution of the above system. It gives the discrete logarithms of elements of $\FF$ (maybe not all). Note that the number of non zero coefficients $e_f$ per relation is very low, let's say a constant $\lambda$. Then it's possible to compute this phase with sparse linear algebra: if $|\FF|$ is an upper bound on the number of unknowns, then algorithms such as Wiedmann algorithm find solutions in $O(\lambda |\FF|^2)$ operations.
\item { \bf Individual discrete logarithm.} Also known as the descent phase, here the aim is to find a relation only involving the target and factor base elements:
\[h\prod_{f \in \FF} f^{e_f} = 1.
\] Since the discrete logarithms of $\FF$ are known, we can reconstruct $\log_g h$ by computing 
$\log_g h = - \sum_{f \in \FF} {e_f}. \log_g f $.
\end{itemize}
Because of potential dependencies, it's not easy to see when we have enough relations,
but in practice having a few more relations than factor base elements
is usually sufficient.

To turn this into a provable algorithm, it is not only necessary to prove
that the finite field representation can be constructed, it is also
required to change the relation generation and linear
algebra. Essentially, one follows the approach
of~\cite{Pomerance1987}, which consists in decomposing plenty of
elements of the form $g^ah^b$ over the factor base. Linear algebra can
then be used to eliminate the contribution of the factor base by
combining the equations which leads to a random identity of the form
$g^Ah^B=1$ from which the discrete logarithm can be recovered.
Unfortunately, this is much less practical. In
particular, there no longer is an individual logarithm phase in this
approach and the full computation has to be restarted from scratch for every
discrete logarithm computation in the same field.

Assuming that the field representation exists, this makes the
computation of discrete logarithms provable. For Frobenius
representation algorithms, this type of provable approach (assuming that the
field representation is given) has been studied in~\cite{GKZ2018,DBLP:journals/iacr/KleinjungW18,DBLP:journals/iacr/GologluJ18}.

\paragraph{Overview of the goal with elliptic representations.}
Let's $\F_{p^k}$ be the target finite field in which we want to
compute discrete logarithms. To simplify exposition, we assume  
$p \geq 5$, since the equation of the curves needs  to be chosen
differently in characteristic $2$ or~$3$. 
We summarize our construction as follows:
\begin{enumerate}
\item {\bf Representation.} Create an elliptic curve $\E$ over $\F_q$ ($q$ being a power of $p$) such that:
\[ \# (\E/\F_q) = \mu k,
\] with $\mu$ a natural integer.  If $k$ is square, there necessarily
exists a $k$-torsion point $P_1 \in \E$. Otherwise, for every prime
$\ell$ dividing $k$, let $\ell^{e_\ell}$ be the largest power of
$\ell$ dividing $k$.  For each $\ell$ with $e_\ell>1$, it might be necessary to change $\E$
to an isogenous curve by applying a sequence of $\ell$-isogenies in
order to guarantee that a point of order $\ell^{e_\ell}$ exists. These
changes to $\E$ suffice to guarantee the existence of $P_1$.

Finally, find a point $F \in \E$
such that:
 \[\pi(F) = F +P_1,
 \] where $\pi$ is the Frobenius action in the field $\F_q$. In
 particular, if we note $F = (\theta, \tau)$ then we can write
 $\F_{q^{k}}$ as $\F_{q}[\theta, \tau]$.
 
 This almost gives the desired representation of the target
 field. Indeed, $\F_{q}[\theta]$ is either $\F_{q^{k}}$ or $\F_{q^{k/2}}$.
In the sequel, we assume that $\F_{q}[\theta]=\F_{q^k}$, multiplying
$k$ by $2$ if necessary.

Then, since $\F_{p^k}$ is a subfield of both $\F_{q^k}$  or
$\F_{q^{2k}}$,  we see that computing discrete logarithms in
$\F_{q}[\theta]$ is sufficient to achieve the desired goal.

\item {\bf Commutative Diagram.} We now define the full representation
  we want to use from a curve $\C$ in 3 dimensions obtained as the
  image of the following rational map:
$$
\begin{array}{rcc}
\Phi : \E &\mapsto& \GFbar{q}^3\\
Q & \mapsto & (x_{Q-P_1}, x_Q, x_{Q+P_1})
\end{array}
$$
At first, this might seem to be a strange model of an elliptic curve.
However, the intuition is that, with this model, the image of $F$ is a
point with a really useful property: taking the Frobenius of one of
its coordinate leads to the following one. In other words, if $\Phi(F)$ is
seen as a point of $\C$ in the affine space $\F_q[U,V,W]$ then it lies
on the intersection of the surfaces defined by the two equations
$U^q=V$ and $V^q=W$. This property is at the core of our
method for creating relations.  Starting from $A$ and $B$ two
polynomials in $\F_q[U,V]$, we construct two big polynomials:
\[A^qB - A B^q = \prod_{\alpha \in \mathbb{P}_1(\F_q)} (A-\alpha B)
\] in one hand, and:
\[A(V,W)B(U,V) - A(U,V) B(V,W) 
\] in the other hand. Each polynomial can be considered as an element of the function
field $\F_q(\C)$. Writing the divisor associated to each side, we can
write down an equality between the image of each divisor in
$\F_{q^k}$. For polynomial themselves, the image is simply obtained by
evaluation at $\Phi(F)$. This can be extended to divisors as explained in
Section~\ref{sec-maptofield}. The equality of the two sides comes from the Frobenius
relations between the coordinates of $\Phi(F)$.

\item {\bf Relation collection.} We sieve on pairs of polynomials $(A,B)$ such 
that $A = g_1 + \alpha g_3$ and
$B= g_1+\beta g_2 + \gamma g_3$ where $\alpha, \beta, \gamma \in \F_q$ 
and $g_1, g_2$ and $g_3$ are given polynomials constructed by linear combination of 
the monomials $U, V, UV$ and $1$.
Set the factor base $\FF$ as all the divisors of $\E$ with height at most~$3$.
On one side  $\prod_{\alpha \in \mathbb{P}_1(\F_q)} (A-\alpha B)$ will always lead to divisors
that can be written as sum of divisors of $\FF$, and on the other 
side $A(V,W)B(U,V) - A(U,V) B(V,W) $ have a low enough height so that the probability that the 
related divisor $\mathcal{D}$ splits in factor base elements is high
enough to get as many relations we want.

Since we only need three degrees of freedom from the four monomials $U, V, UV$ and $1$, we choose $g_1, g_2$ and
$g_3$ all going through a common point. This nicely reduces the
degree of the divisors appearing in the decomposition of the terms
$A-\alpha B$. This is essential in making the probability of success
during the relation collection phase good enough.

\item {\bf Linear algebra and individual logarithm. } Thanks to the
  action of Frobenius, we can reduce the size of factor base by a
  factor $k$.  In other words, this reduces the effective size of
  the factor base $\FF$ to $O(q^3/k)$. As a consequence, the cost of
  the sparse linear algebra, with $O(q)$ entries per equation, is of
  $O(q^7/k^2)$ arithmetic operations. Note that, when $k$ is chosen
  close to $q$, this matches the $O(q^5)$ asymptotic complexity
  obtained for this step in~\cite{DBLP:conf/asiacrypt/JouxP14}.

\item {\bf Descent Phase.} Finally, we need a descent phase to
  conclude. We give the necessary tools to adapt  existing
  methods in this context.

\end{enumerate}

\paragraph{Outline.} Section~\ref{preliminaries} gives algebraic preliminaries for this work.

In this paper, we focus on the algorithmic aspects and describe our
practical elliptic Frobenius algorithm, that can be helpful to fully
understand the general idea. Being aimed at practicality, this
algorithm is heuristic. From a performance analysis point of view, our
heuristic approach almost achieves the same efficiency as the best pre-existing
practical algorithm for DLP in small characteristic finite fields. {\it Almost},
because there is a glitch in the analysis of the fast computation of
the extended factor base. However, despite this glitch, we were able
to implement and use the elliptic representation approach to compute logarithms of an extended
factor base for the finite field $\F_{3^{1345}}.$

We highlight the heuristics we use as far as we can, in order to
clarify the difference with the provable algorithm of~\cite{KW19}
or~\cite{GuidoLido2}. 

In Section~\ref{representation} we give our variation on the
representation of the target finite field while
Section~\ref{collection} details how to get relations. Finally
Section~\ref{extended} deals with factor base extension and with the
individual logarithms phase.

\section{A Refresher on the Function Field Sieve Machinery}
\label{preliminaries} 
Many concepts used here originate from the Function
Field Sieve (FFS) algorithm~\cite{BLP:journals/iandc/AdlemanH99}. The aim of this first section is not to
describe FFS itself, but to describes these
concepts in a slightly more general form than the original description
of Adleman and Huang article.

\subsection{Algebraic preliminaries}

Let $\EE=\GF{q}$ denote a finite field.  Let $\C$ be a non-singular
curve in the $n$-dimensional projective space $\PP_n(\GFbar{q})$
defined over $\EE$ and $\Fr$ denote the Frobenius map on
$\PP_n(\GFbar{q})$. We let $\EE(\C)$ denote the function field of
$\C$ over $\EE$. More details can be found in~\cite{DBLP:reference/crc/2005ehcc} if 
needed. A {\it discrete valuation} on $\EE(\C)$ is a map
$\valua$ from $\EE(\C)$ to $\Z$ such that for all $x,y \in \EE(\C)$ we
have:
\begin{enumerate}
\item $\valua(xy)=\valua(x)\,\valua(y);$
\item $\valua(x+y)\geq\min(\valua(x),\valua(y));$ 
\item $\valua(x+y)=\min(\valua(x),\valua(y))$ when $ \valua(x)\neq \valua(y)$.
\end{enumerate}
We define an equivalence relation between valuations by saying that
two valuations $\valua$ and $\valua'$ are equivalent whenever there
exists a non zero rational constant $\alpha$ such that  for all $x
\in \EE(\C)$, we have $\valua'(x)=\alpha\,\valua(x).$ 
We recall that a {\it place} of $\EE(\C)$ is an
equivalence class of discrete valuations of $\EE(\C)$ which are trivial
on $\EE$. The set of places of $\EE(\C)$ is denoted by
$\Sigma_{\EE(\C)}$. In every place $\pcal$, there exists a unique
valuation whose value group is $\Z$, it is called the \textit{normalized valuation} of $\pcal$ and denoted $\valua_{\pcal}.$

We recall that, for a non-singular curve $\C$, there is a one-to-one
correspondance between places of $\EE(\C)$ and Galois orbit of points
on $\C$. The {\it degree} of a place $\pcal$ is the number of points in the
corresponding orbit, we denote it by $\Deg(\pcal)$.

The {\it divisor group} $\Div(\C)$ of $\C$ (over $\EE$) is defined as the free
abelian group over $\Sigma_{\EE(\C)}$. An element $D$ of $\Div(\C)$
is expressed as:
$$
D=\sum_{\pcal \in \Sigma_{\EE(\C)}}n_{\pcal}(\pcal),
$$
where each $n_{\pcal}\in \Z$ and $n_{\pcal}=0$ for all but finitely many
places $\pcal$.
Since each place corresponds to a Galois orbit of points, a divisor
$D$ can also be given in the alternative form:
$$
D=\sum_{P \in \C/\GFbar{q}}n_{P}(P),
$$
where each $n_{P}\in \Z$, $n_{P}=0$ for all but finitely many
points and $n_P=n_Q$ if $P$ and $Q$ belong to the same Galois orbit of
points. A divisor $D$ is said to be {\it prime} when $D=(\pcal)$ for a place $\pcal \in \Sigma_{\EE(\C)}$.

The {\it degree} of a divisor $D$ is defined as:
$$
\Deg(D)=\sum_{\pcal \in \Sigma_{\EE(\C)}}n_{\pcal}\Deg(\pcal)=\sum_{P \in \C/\GFbar{q}}n_{P}.
$$
In this paper, a degree-$0$ divisor that is the difference between a prime divisor 
  and the right number of times the point at infinity $\OO$ is defined as an 
  \textit{elementary divisor}. In particular, any elementary divisor associated to a point  
  $Q \in \C/\F_{q^d}$ is a divisor of the form: 
\[ \sum_{i=0}^{d-1}\Fr^i(Q)-d(\OO).\]

A divisor $D$ is called {\it effective} when $n_{\pcal}\geq 0$ for all
$\pcal$. Any divisor $D$ can be uniquely written as a difference of
two effective divisors in the form $D=D_0-D_{\infty}$, where:
$$
D_0 = \sum_{\begin{array}{c}\pcal \in 
                \Sigma_{\EE(\C)}\\n_{\pcal}\geq
              0\end{array}}n_{\pcal}(\pcal) \quad \mbox{and} \quad
D_{\infty}= \sum_{\begin{array}{c}\pcal \in 
                \Sigma_{\EE(\C)}\\n_{\pcal}<0\end{array}}-n_{\pcal}(\pcal).
$$

The degree map from $\Div(\C)$ to $\Z$ is a group morphism. Its kernel
is denoted $\Div_0(\C)$ and called the group of degree-$0$ divisors of
$\C$, it is a subgroup of  $\Div(\C)$.

We define the map\footnote{Other authors often use the notation div but we do not want the map to be mistaken with the group.}  $\Xi$ that sends an element $f \in
\EE(\C)^{*}$ to a divisor in the following way:
\[\begin{array}{cccc}
\Xi:& \EE(\C)^{*}&\mapsto& \Div(\C)  \\
&f&\mapsto&\displaystyle \Xi(f) =\sum_{\pcal \in \Sigma_{\EE(\C)}}\valua_{\pcal}(f)\,\pcal.
\end{array}
\]
A divisor associated to a function in the above way is called a {\it
  principal divisor}. 
  The image of $\Xi$, i.e. the set of all
principal divisors, is denoted $\Princ(\C)$. All principal divisors
have degree $0$ and $\Princ(\C)$ is a subgroup of $\Div_0(\C)$. Every
principal divisor can also be written as a difference of effective
divisors as:
$$\Xi(f)= \Xi(f)_{0}-\Xi(f)_{\infty}.$$

The places (or points) that occur in $\Xi(f)_{0}$ or
$\Xi(f)_{\infty}$ are respectively called the {\it zeroes} or {\it
  poles} of $f$.
Note that for two functions $f$ and $g$ of  $\EE(\C)^{*},$ we have
$\Xi(f)=\Xi(g)$ if and only if there exists an element $\alpha \in
\EE^{*}$ such that $g=\alpha\, f$.

Since $\Princ(\C)$ is a subgroup of $\Div_0(\C)$, we can form the
quotient group, which is called the \textit{Picard group} (or divisor class
group) of $\C$ and denoted  $\Pic_0(\C)$. Two divisors have the same
representative in the Picard group if and only if their difference is principal.

\subsection{\label{sec-maptofield}A tool from FFS:  sending a degree-0 divisor into a finite field}
With these algebraic objects in hand, we can now introduce the main
tool that we need to import from the Function Field Sieve. Let
$\EE=\GF{q}$ be a finite field, $\LL=\GF{q^k}$ be a finite extension
of $\EE$ and $\C$ be as before a non-singular curve
defined over $\EE$. Let $F$ be a point of $\C/\LL$ such that the
coordinates of $F$ generate $\LL$ over $\EE$.
Given an arbitrary element $f\in \EE(\C)^{*}$ which does not have
$F$ as a pole, we can evaluate $f$ at $F$ and obtain a value in~$\LL$. 
As explained in~\cite{BLP:journals/iandc/AdlemanH99}, this process can be generalized from
functions to a large subset of divisors of $\C$. Clearly, since
$\alpha\,f$ and $f$ have the same divisor for any $\alpha\in \EE$, we
need to proceed with care.

First, we define  a map $\Psi$ from $\Princ(\C)$ to $\LL/\EE^*$ defined as
follows:
\[\begin{array}{cccc}
\Psi:& \Princ(\C)&\mapsto&  \LL/\EE^{*}\\
&D&\mapsto& \Psi(D)=f(F), 
\end{array}
\]
where $f$ is an arbitrary  function such that $\Xi(f)=D$. Since the
result is only considered up to multiplication by an arbitrary
constant in $\EE^{*}$, it is independent of the particular choice of
$f$.

To generalize to more divisors, we now consider a degree-$0$ divisor
$D$ together with an integer $h\in \N^{*}$ such that $hD$ is principal
and $h$ is coprime to the order of $\LL^*/\EE^{*}=(q^k-1)/(q-1)$. For
such a divisor, we extend the definition by letting:
$$
\Psi(D)=\Psi(hD)^{1/h}.
$$
In the terminology of~\cite{BLP:journals/iandc/AdlemanH99}, we evaluate at $F$ the
``surrogate'' 
function that we have associated to $D$ thanks to the
multiplication by $h$. Note that replacing $h$ by any other $h'$
satisfying the conditions does not change the value of $\Psi(D)$.
Furthermore, remark that if $\Psi$ is defined on $D$ and $D'$, it is
also defined on $D+D'$ and $\Psi(D+D')=\Psi(D)\cdot
\Psi(D')$. Similarly, if $\Psi(D)$ is defined and non-zero, then
$\Psi(-D)=1/\Psi(D)$. 

From a computational point of view, if $D$ has a small support, then
$\Psi(D)$ can be efficiently computed by using Miller's algorithm to
compute the evaluation at $F$ of the function corresponding to
$hD$. See Section~\ref{diagram} for more details.

In the Function Field Sieve, this tool is used as part of the
commutative diagram that underlies the construction of multiplicative
relations. However, it is not used directly in the algorithm, only in
its correctness proof. Similar, in our elliptic representation
algorithm, the map $\Psi$ is not really necessary to perform
computations. However, having it at our disposal gives a very useful
tool for checking the correctness of relations, thus helping to remove
undesirable implementation bugs.

\section{Representation of the Target Finite Field}
\label{representation}

Let $q$ and $k$ respectively be a prime power and the extension degree
of the field. Our aim is to compute discrete logarithms in it. Let's
write $p$ its characteristic. Since we want to define the extension
as $\F_q[x_P]$, with $x_P$ the abcissa of some point on an elliptic
curve, while the construction naturally writes it as $\F_q[x_P,y_P]$,
there is a small risk of producing a subextension (missing a last
quadratic extension) when $k$ is even. If this happens, it suffices
to replace $k$ by $2k$ during the initial choice of the elliptic
basis. As a consequence, we can safely ignore this point in the sequel.

\subsection{Choosing the elliptic basis}
The representation step of our algorithm starts by forming an elliptic
curve over a small extension of $\F_{q}$ with cardinality a multiple
of $k$.
The following result explicits bounds with respect to $q$ and $k$
for both the extension degree and the multiplying factor.

\begin{theorem}
\label{iliade}
If $q$ and $k$ be a prime power and a positive integer, then there exist
$\mu$ and $\nu$ two integers and
an elliptic
curve over~$\F_{q^\mu}$ with cardinality $\nu k$
such that $\mu \leq  \lceil \log(k^2/4)/\log q\rceil +1$.
\end{theorem}

\begin{proof}
  Let $p$ denote the characteristic of $\F_{q}$.
  We can perform following case by case analysis:
\begin{enumerate}
\item
First, let us assume that $k <2 \sqrt{q}$. Thus
there exist at least two multiples $\nu k$ and 
$(\nu +1) k$ of $k$ in the Hasse interval
$\left]q+1-2 \sqrt{q}, q+1+2 \sqrt{q} \right[$.
Furthermore, we may assume that $\nu k<q+1.$ 

If $p$ divides $k$ then $p$ cannot divide the trace $t=q+1- \nu k$.
Let us recall now the following result:

\begin{corollary}[of Waterhouse's theorem~\cite{Waterhouse1969}]
For each value of the characteristic~$p$, for any extension degree~$n$
and for every integer $t$ in $\left]-2 \sqrt{p^n}, 2 \sqrt{p^n} \right[$
such that $t \not \equiv 0 \mod p$, there exists an elliptic curve over
$\F_{p^n}$ whose number of rational point is exactly ${p^n}+1-t$.
\end{corollary}

The reader can for example find a proof in~\cite{Ughi}. Note that we can run into
some cases where the characteristic does divide the trace and yet such that
there exists such a curve.
Theorem~$4$ in~\cite{Ughi} gives the exhaustive list of
these special cases.

Back to our discussion of the case where $p$ divides $k$, we see that,
Waterhouse's theorem yields the existence of an elliptic curve over $\F_{q}$ with cardinality
$\lambda k$, for all values $\lambda k$ in the Hasse interval. In
particular, there exist a curve of cardinality $\nu k$.

If $p$ does not divide $k$ then two sub-cases occur. Either $q+1 - \nu k$
is not a multiple of~$p$, and Waterhouse's theorem permits to conclude again that there exists
an elliptic curve over $\F_{q}$ with cardinality $\nu k$,
or it is. In the second case, $p$ doesn't divide  $(\nu+1) k$ and we
obtain a curve with that cardinality.

Either way,  when  $k <2 \sqrt{q}$ we always find a curve over the base field $\F_{q}$.

\item If $k \geq 2 \sqrt{q}$ then there is no guarantee of the
  existence of two multiples of $k$ in the Hasse interval.  Thus,
  unless we are lucky, we need to increase the size of the finite
  field to get a larger interval. Let $\mu$ be the smaller integer
  such that $k <2 \sqrt{q^\mu}$. Applying the previous case on this
  extended field, we see that there always exists an elliptic curve
  over $\F_{{q^\mu}}$ with cardinality $\nu k$ such that
  $\nu k < {q^\mu}+1+2 \sqrt{p^\mu}$. To conclude, notice that the
  additional extension degree $\mu$ that we need is (at most) equal to
  the ceiling of $ \log(k^2/4)/\log q$.

\end{enumerate}
\end{proof}

Once we have found $\E$, it allows us to define the finite field
$\F_{q^{\mu k}}$, where $\mu$ is the extra extension degree needed to
find $\E$. To lighten notations, we assume without loss of generality
that $\mu=1$. Indeed, it suffices to redefine a new value for $q$
equal to the previous value $q^\mu$. Thanks to the upper bound on the
extension degree, we see that it does not affect the quasi-polynomial time complexity
of the algorithm.

We further assume
that $\E$ contains a point in $\F_q$ of order $k$. If necessary, apply
low-degree isogenies to the initial curve $\E$ until a suitable one is
obtained. 

Then construct a point $F$ whose coordinates in the algebraic closure
satisfy:
 \[\pi(F) = F +P_1,
 \] where $\pi$ is the $q$-th power Frobenius action. Write the
 coordinates of $F$ as $(\theta, \tau).$ From~\cite{DBLP:journals/ffa/CouveignesL09},
 we know that $\F_{q^k}=\F_q[\theta,
 \tau]$. Furthermore, from our assumption on $k$, we have $\F_{q^k}=\F_q[\theta]$.

Note that, when focusing on the practical variation of the algorithm, it
is important to have $k$ as large as possible compared to $q$.
The above proof only guarantees that $q=O(k^2)$, however, in the best
cases we can have $q=O(k)$.

\subsection{Representing the curve with a different model}
\label{construction}

We introduce here a model that represents elliptic curves in the
three-dimensional affine space $\F_q[U,V,W]$. To construct this new
model, we start from a curve $\E$ together with a $k$-torsion point
$P_1$ of~$\E/\F_q$. We emphasize that $P_1$ has all its coordinates in
the field~$\F_q$. Let $(x_1,y_1)$ denote the coordinates of $P_1$ and
$(x_\ell, y_\ell)$ the coordinates of $P_\ell=\ell\,P_1$ for
$\ell \in [2, \cdots, k-1]$.

\paragraph{Adding some structure.}
The idea is to create a new model of $\E$ in which we artificially
inject extra desirable properties. Namely, for any point
$Q \in \E/\bar{\F}_{q}$, we represent it by the triple of abcissae of
the points $Q-P_1$, $Q$ and $Q+P_1$ , on the one hand and $-Q, P_1$, $Q-P_1$ on the
other hand (see Figure~\ref{fig frob2}).

In this model, there is an easy way to add the $\pi(F) = F +P_1$
constraint of the Couveignes and
Lercier~\cite{DBLP:journals/ffa/CouveignesL09} construction of
elliptic bases. Graphically, this is shown on Figure~\ref{fig frob3}.

Indeed, for point $F$ we see that the triple of coordinates is
$(x_{\pi^{-1}(F)},x_{F},x_{\pi(F)})$. Furthermore, for $\pi(F)$ the
triple is $(x_{F},x_{\pi(F)}, x_{\pi^2(F)}).$ As a consequence, the
first two coordinates can be obtained by a simple shift.

Furthermore, it is possible to recover the missing coordinate of
$\pi(F)$ from the first two, in a way similar to Montgomery's ladder
technique.

\paragraph{Formal definition of $\C$.}
Now that we have captured our intuition, we give the equations of the
curve in the new model. First, we recall the definition of the third Semaev
polynomial $S_3$, it is an irreducible and symmetric polynomial of
degree-$2$ in $\F_q[U,V,W]$.
Furthermore, for any triple of points
$Q_1 = (x_{Q_1}, y_{Q_1}), Q_2 = (x_{Q_2}, y_{Q_2})$,
$Q_3 = (x_{Q_3}, y_{Q_3}) \in \E(\bar \F_q) \setminus \{ \OO\}$, we
have:
\[S_3(x_{Q_1},x_{Q_2},x_{Q_3}) =0 \Leftrightarrow \exists (e_1, e_2, e_3) \in \{ -1,1\}^3, e_1 Q_1+e_2Q_2+e_3Q_3= \OO.
\]

We use this polynomial to describe the image $\C$ of the rational map:
$$
\begin{array}{rcc}
\Phi : \E &\mapsto& \GFbar{q}^3\\
Q & \mapsto & (x_{Q-P_1}, x_Q, x_{Q+P_1})
\end{array}
$$

For every point $Q \in \E/\bar{\F}_{q}$, we use $S_3$ to rewrite the
three simple identities $(Q-P_1) - Q + P_1= \OO$, $Q-(Q+P_1) + P_1= \OO$ and
$(Q-P_1) - (Q + P_1) + P_2= \OO$. This shows that 
the point $(x_{Q-P_1}, x_{Q},x_{Q+P_1})$ is a common root of  the polynomials $S_3(U,V,x_1)$, $S_3(V,W,x_1),$ and $S_3(U,W,x_2)$.
Yet, the
variety defined by these $3$ equations
contains several components. 
One of the components is irreducible and
has dimension $1$ whereas the others correspond to extraneous points. 
This irreducible component is the curve~$\C$ given by the equations:
\[S_3(U,V,x_1)=0,\quad \frac{S_3(U,V,x_1)-S_3(V,W,x_1)}{U-W}=0, \quad S_3(U,W,x_2)=0.\]
In fact $\C$ is isomorphic to the initial elliptic curve
$\E$.  For more details,
see Appendix~\ref{appendix}.
 \begin{figure}
\begin{center}
\begin{tikzpicture}[domain=-3:3, scale=1]

\begin{scope}
    \clip (2,0) rectangle (4,4);
     \draw[domain=2:3,smooth,variable=\x,blue]
     plot({\x},{sqrt((\x^3-4*\x))}) node[anchor=-180]{\footnotesize $\E$}; 

\end{scope}
\begin{scope}
    \clip (-2,0) rectangle (2,4);
     \draw[domain=-2:0,smooth,variable=\x,blue] plot({\x},{sqrt((\x^3-4*\x))}); 
\end{scope}
\begin{scope}
    \clip (2,-4) rectangle (4,4);
     \draw[domain=2:3,smooth,variable=\x,blue] plot({\x},{-sqrt((\x^3-4*\x))}); 
\end{scope}
\begin{scope}
    \clip (-2,-4) rectangle (2,4);
     \draw[domain=-2:0,smooth,variable=\x,blue] plot({\x},{-sqrt((\x^3-4*\x))}); 
\end{scope}  

\draw[->, line width=0.5pt] (-3,0) -- (3,0);
\draw[->, line width=0.5pt] (0,-4) -- (0,4);
\coordinate (P1) at (-1.8,- {sqrt (-(1.8)^3+1.8*4)});
\coordinate (Q) at (2.7, - {sqrt ((2.7)^3-2.7*4)});
\coordinate (-Q) at (2.7,  {sqrt ((2.7)^3-2.7*4)});
\coordinate (Q-P1) at (-0.1, 0.4);
\coordinate (aux) at (-0.75, -1.6);
\coordinate (Q+P1) at (-0.75, 1.6);
\draw [color=red] (P1) -- (Q);
\draw [color=red] (P1) -- (-Q);

\fill [black] (P1) circle (1pt) node[anchor=-20]{\footnotesize $P_1$};
\fill [black] (Q) circle (1pt) node[anchor=-180]{\footnotesize $Q$};
\fill [black]  (Q-P1) circle (1pt) node[anchor=170]{\footnotesize $Q-P_1$};
\fill [black]  (aux) circle (1pt) ;
\fill [black]  (-Q) circle (1pt) ;
\fill [black]  (Q+P1) circle (1pt) node[anchor=-90]{\footnotesize $Q+P_1$};

\coordinate (Qa) at (2.7, 0);
\coordinate (Q-P1a) at (-0.1,0);
\coordinate (Q+P1a) at (-0.75,0);

\draw [color=green, dashed] (Q) -- (Qa);
\draw [color=green, dashed] (Q-P1) -- (-0.1,0);
\draw [color=green,dashed] (Q+P1) -- (-0.75,0);
\fill [black] (Qa) circle (1pt) node[anchor=-90]{\footnotesize $x_Q$};
\fill [black] (Q-P1a) circle (1pt) node[anchor=120]{\footnotesize $x_{Q-P_1}$};
\fill [black]  (Q+P1a) circle (1pt) node[anchor=60]{\footnotesize $x_{Q+P_1}$};
\draw [->, color=green] (Q-P1a) to [out=90, in=90] (Qa);
\node[green] at (1.3,1) {$\pi$} ;
\draw [->, color=green] (Qa) to [out=-90, in=-90] (Q+P1a);
\node[green] at (1,-1.3) {$\pi$} ;

\end{tikzpicture}
\end{center}
\caption{\label{fig frob2}Frobenius action on abscissae as we would like.}
\end{figure}
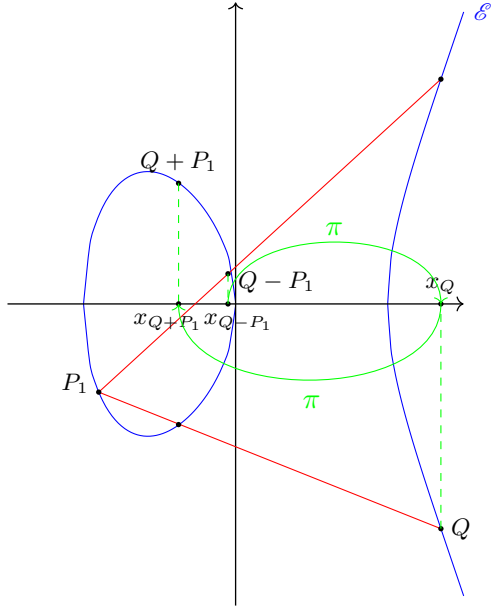

 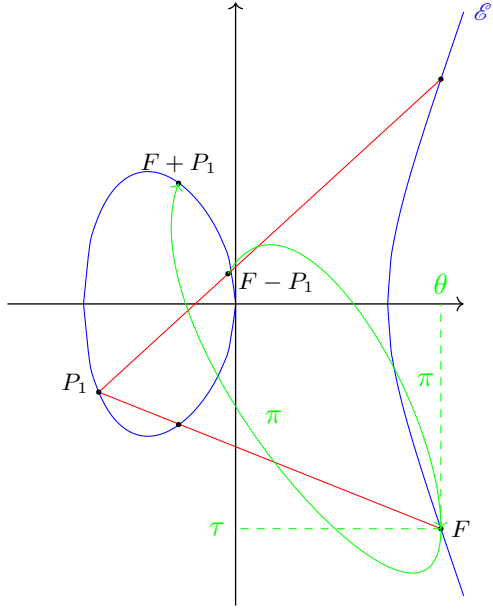
\begin{figure}
\begin{center}
\begin{tikzpicture}[domain=-3:3, scale=1]

  \begin{scope}
    \clip (2,0) rectangle (4,4);
     \draw[domain=2:3,smooth,variable=\x,blue]
     plot({\x},{sqrt((\x^3-4*\x))}) node[anchor=-180]{\footnotesize $\E$}; 

\end{scope}
\begin{scope}
    \clip (-2,0) rectangle (2,4);
     \draw[domain=-2:0,smooth,variable=\x,blue] plot({\x},{sqrt((\x^3-4*\x))}); 
\end{scope}
\begin{scope}
    \clip (2,-4) rectangle (4,4);
     \draw[domain=2:3,smooth,variable=\x,blue] plot({\x},{-sqrt((\x^3-4*\x))}); 
\end{scope}
\begin{scope}
    \clip (-2,-4) rectangle (2,4);
     \draw[domain=-2:0,smooth,variable=\x,blue] plot({\x},{-sqrt((\x^3-4*\x))}); 
\end{scope}  
\draw[->, line width=0.5pt] (-3,0) -- (3,0);
\draw[->, line width=0.5pt] (0,-4) -- (0,4);
\coordinate (P1) at (-1.8,- {sqrt (-(1.8)^3+1.8*4)});
\coordinate (Q) at (2.7, - {sqrt ((2.7)^3-2.7*4)});
\coordinate (-Q) at (2.7,  {sqrt ((2.7)^3-2.7*4)});
\coordinate (Q-P1) at (-0.1, 0.4);
\coordinate (aux) at (-0.75, -1.6);
\coordinate (Q+P1) at (-0.75, 1.6);
\draw [color=red] (P1) -- (Q);
\draw [color=red] (P1) -- (-Q);

\fill [black] (P1) circle (1pt) node[anchor=-20]{\footnotesize $P_1$};
\fill [black] (Q) circle (1pt) node[anchor=-180]{\footnotesize $F$};
\fill [black]  (Q-P1) circle (1pt) node[anchor=170]{\footnotesize $F-P_1$};
\fill [black]  (aux) circle (1pt) ;
\fill [black]  (-Q) circle (1pt) ;
\fill [black]  (Q+P1) circle (1pt) node[anchor=-90]{\footnotesize $F+P_1$};

\coordinate (Qa) at (2.7, 0);
\coordinate (Qb) at (0, - {sqrt ((2.7)^3-2.7*4)});
\coordinate (Q-P1a) at (-0.1,0);
\coordinate (Q+P1a) at (-0.75,0);

\draw [->, color=green] (Q-P1) to [out=60, in=90] (Q);
\node[green] at (2.5,-1) {$\pi$} ;
\draw [->, color=green] (Q) to [out=-90, in=-110] (Q+P1);
\node[green] at (0.5,-1.5) {$\pi$} ;
\draw [dashed, color=green] (Q) -- (Qa) node[above] {$\theta$};
\draw [dashed, color=green] (Q) -- (Qb) node[left] {$\tau$};

\end{tikzpicture}
\end{center}
\caption{\label{fig frob3}Frobenius action on the point $F$.}
\end{figure}

\paragraph{Point in $\E$ with coordinates in the target finite field.}
Let $F \in \E$ be such that:
 \begin{equation}
 \label{frobenius}
\pi(F)=F+P_1.
\end{equation}

\begin{lemma} 
\label{dansleboncorps}
If $F \in \E/\bar{\F}_q$ is such that $\pi(F)=F+P_1$ then
$F \in \E/\F_{q^k}$. Furthermore, letting $(\theta, \tau)$ denote the
coordinates of $F$, we have $\F_q[\theta, \tau]=\F_{q^k}$. In
particular, there exist at least $k$ rational points verifying the same property.
\end{lemma}

\begin{proof}
Let us  compute $\pi^k(F)=\pi^{k-1}(F +P_1)=\pi^{k-1}(F) +P_1 = \cdots =F+k P_1 $.
We know that $P_1$ has precisely order $k$ hence  $\pi^k(F)=F$. Besides, we note
that any point $\pi^i(F)$ for $i=1, \cdots , k-1$ satisfies Equation~\eqref{frobenius} too.
\end{proof}

As already mentioned, having possibly doubled the value of $k$ in
construction, we may assume that $\F_q[\theta]=\F_{q^k}.$ 
in our model of curve, the point $F$ is determined by the fact that
$S_3(\theta, \theta^q,x_1)=0$.
The abscissa~$\theta$ of~$F$ can thus be determined as a root of this
polynomial. Note that the choice of the ordinate~$\tau$ gives an
orientation on the direction of the Frobenius action. We choose the
letter $F$ to name this point as a mnemonic to remind that it
represents our target Finite Field and that it has a special relationship
with the Frobenius map.

\subsection{Commutative diagram}
\label{diagram}

From this, we derive the commutative diagram of Figure~\ref{fig:EllipticCD} which
serves as the basis for our elliptic Frobenius representation
algorithm. Note that the commutative diagram is above $\F
_{q^{k}}/\F_q^{*}$, as a consequence, our algorithm doesn't compute
the part of the discrete logarithm corresponding to
$\F_q^{*}$. However, this field is so small that this missing part can
be easily obtained.

\begin{remark}
This diagram could be simplified by removing references to
the function field~$\F_q(\E)$ and computing divisors on~$\C$
directly. However, when using standard computer algebra tools, it is
much simpler to work on divisors with the Weirstrass equation of $\E$.
\end{remark}

\begin{figure}[h]
\centering
\begin{tikzpicture}[commutative diagrams/every diagram]
  \node (P0) at (0cm,0cm) {$\begin{array}{c}
\F_q \left[ U,V \right] \times   \F_q \left[ U,V \right] \\
(A,B)
\end{array}$};
  \node (P1) at (-5cm,-2cm) {$ \begin{array}{c}
\F_q \left[ U,V,W \right] \\
\displaystyle 
A^q\,B-A\,B^q  
\end{array}$};
  \node (P1bis) at (-2cm,-2.45cm) {$ \begin{array}{c}
\displaystyle =\prod_{\alpha\in\PP_1(\F_q)}(A-\alpha\,B)
\end{array}$};

  \node (P2) at (5cm,-2cm) {$ \begin{array}{c}
\F_q \left[ U,V,W \right] \\
\displaystyle 
A(V,W)\,B(U,V)-A(U,V)\,B(V,W)
\end{array}$};
  \node (P3) at (-5cm,-4cm) {$ \displaystyle \F_q \left[ U,V,W \right]/I$};
  \node (P3bis) at (-3cm,-4cm) {$ \displaystyle =\F_q[\C]$ };
  \node (P4bis) at (2.5cm,-4cm) {$ \displaystyle \F_q[\C] =$};
  \node (P4) at (5cm,-4cm) {$ \displaystyle \F_q \left[ U,V, W \right]/I$};
 \node (P5) at (-5cm,-6cm) {$ \displaystyle \F_q(\E) $};
 \node (P5bis) at (-5cm,-8cm) {$ \displaystyle \Princ(\E) \subset Div^0(\E)$};
   \node (P6) at (5cm,-6cm) {$ \displaystyle \F_q(\E) $};
     \node (P6bis) at (5cm,-8cm) {$ \displaystyle \Princ(\E) \subset Div^0(\E)$};
  \node (P7) at (0cm,-10cm) {$  \displaystyle \F _{q^{k}}/\F_q^{*}$};
  \path[commutative diagrams/.cd, every arrow, every label]
    (P0) edge node {} (P1)
    (P0) edge node {} (P2) 
    (P1) edge node {$i$} (P3) 
    (P2) edge node {$i$} (P4) 
    (P3) edge node {$\Phi ^*$} (P5) 
    (P4) edge node {$\Phi ^*$} (P6) 
        (P5) edge node {$\Xi$} (P5bis) 
    (P6) edge node {$\Xi$} (P6bis) 
    (P5bis) edge node {$\Psi$} (P7) 
    (P6bis) edge node {$\Psi$} (P7) ;
\path[commutative diagrams/.cd]
 (P5bis) edge[dashed] node[below] {\footnotesize\begin{tabular}{c} Factor base elements\\
  are there.
  \end{tabular}}(P6bis) ;
\end{tikzpicture}
\caption{\label{fig:EllipticCD} Commutative diagram of our algorithm.}
\end{figure}
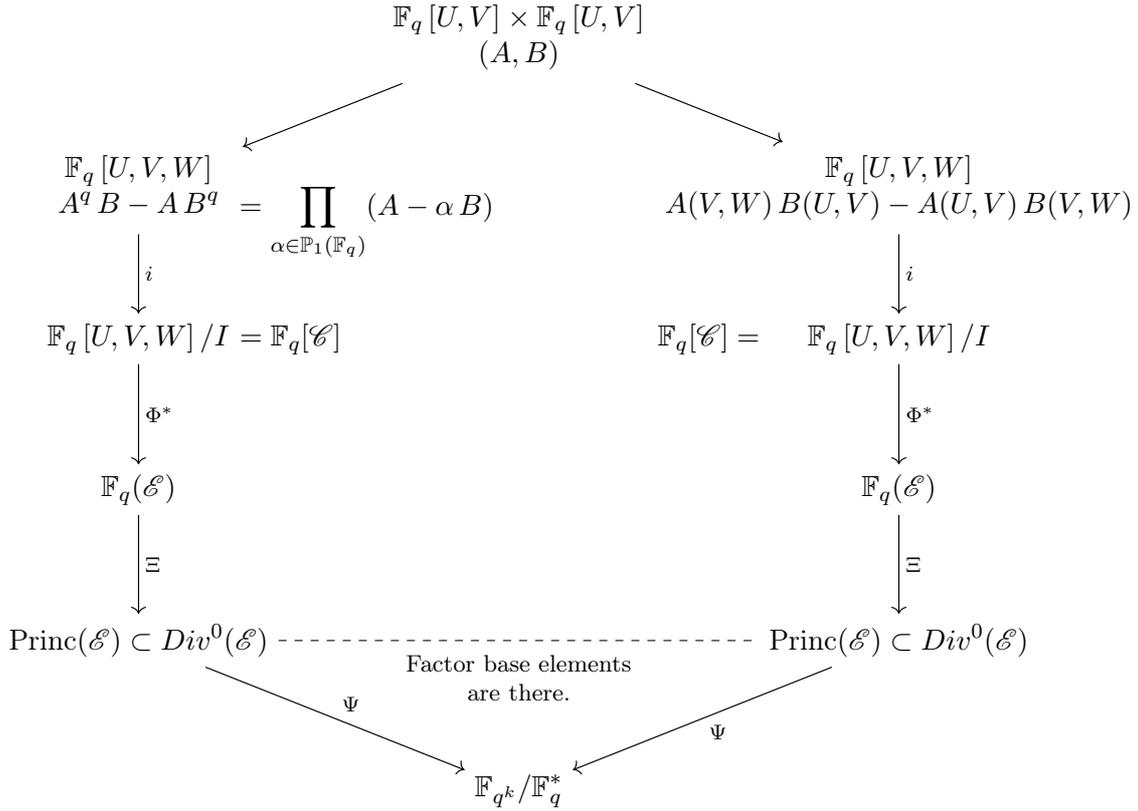

\paragraph{Explicit maps to $\F_{q^k}$ based on Miller algorithm.}

\noindent The first two maps of the diagram
are explicit and the two following ones are canonical
injections. $\Phi^*$ is given in Appendix and $\Xi$ is as defined in
Section~\ref{preliminaries}. In addition, we let $\Psi$ denote a multiplicative group
morphism that sends elements of $ \Princ(\E)$  to
$\F_{q^k}$. 
Yet, only defining $\Psi$ for principal divisors is not sufficient, since
in the relation collection phase we need first to factor divisors in~$ \Princ(\E)$
into elementary divisors before descending them into the finite field. Keep
in mind that elementary divisors have no reason to be principal.

For any divisor in $\Princ(\E)$
the first thing to do is to decompose it into 
elementary divisors in~$Div^0$.  Then, we note that since we wish to construct a group 
morphism that sends elements of  $Div^0 $ to
 $\F_{q^k}/\F_q^{*}$,
it suffices to describe this morphism for any of these elementary divisors.
Let us consider:
$$
\DD_e=\sum_{i=1}^{d}\pi^i(Q)-d(\OO),
$$ 
where $Q\in \E/\F_{q^d}$ is one of the conjugate points in the degree-$d$ place.
Fix a maximum degree $D$ for the places we consider and
let $N_D$ be the least common multiple of the cardinalities of $\E$
over each finite field $\F_{q^d}$ with $1\leq d \leq D$.

From this, we see that $N_D\DD_e$ is a principal divisor, thus there
exists a function $f_{\DD_e}$ in the variables $X$ and $Y$ unique up to multiplication by a constant
in $\F_q$ such that $\mbox{div}(f_{\DD_e})=N_D\DD_e$. 
We want to use the point $F$ with coordinates in the target finite field to define it.
Since $\theta$~and~$\tau$ are respectively the abscissa and the ordinate of
this point, it seems natural to send $X$ to~$\theta$ and
$Y$ to $\tau$, or, in other words, to evaluate the function on the point~$F$. 
However, since~$f_{\DD_e}$ is only defined modulo a constant in $\F_q$,
the result in the finite field would change depending on the choice of the function.
To annihilate this constant, we have to divide the evaluation on $F$ by the evaluation on $\OO$.
Hence to have a well-defined application $\Psi$ we set
$$ \Psi(\DD_e)=\left(f_{\DD_e}\left(F-\OO\right)\right)^{1/N_D}.$$
However, evaluation at $\OO$ isn't really necessary, since we are only
interested in values in $\F _{q^{k}}/\F_q^{*}$.
As done for bilinear pairings, this can be efficiently computed using
Miller's algorithm~\cite{JC:Miller04}.

\begin{remark}
  In order to raise to the power $1/N_D$ and get a uniquely defined
  value, we need to check that $N_D$ is invertible modulo the order of
  $\F _{q^{k}}^{*}/\F_q^{*}$. This condition needs to be tested for
  all the orders of the curve $\E$ in the extension fields
  $\F_q, \F_{q^2}, \cdots, F_{q^D}.$

  We analyze the condition more precisely in
  Appendix~\ref{rightcurve} and provide a replacement for $\Psi$ in
  the case where $N_D$ cannot be inverted.
\end{remark}

\paragraph{Intuition about $\Psi$ and commutativity of the diagram.}
This definition matches with the following intuitive one. 
To unsure the commutativity of the diagram we need to verify that 
$\Psi(\Xi(\Phi^*(i(A^qB-AB^q))))$ is equal in the finite field
to the element $\Psi(\Xi(\Phi^*(i(A(V,W)B-AB(V,W)))))$. 
Our intuition is that requiring in some sense:
\[U^q=V \qquad \hbox{and} \qquad V^q=W\] would suffice to
prove the commutativity.
We point out that one point of the elliptic curve~$\C$, namely $\Phi(F)$, precisely follows 
this restriction. Indeed, $\Phi(F)$ has abscissa
$x_{F-P_1}=x_{\pi^{(k-1)}(\pi(F-P_1))}=x_{\pi^{(k-1)}(F)} =\pi^{(k-1)}(\theta)$,
ordinate $ x_F=\theta$ and applicate $ x_{F+P_1}=x_{\pi(F)}=\pi(\theta)$. In a nutshell:
\[ \Phi(F)=[\theta^{q^{k-1}},\theta,\theta^q]
\]Thus, in the function field $\F_q(\C)$,  evaluation the functions at
$\Phi(F)$ gives the expected relationship to the Frobenius map. 
As a consequence, evaluation at $F$ after transporting back to the
function field of $\E$ using $\Phi^*$ also gives the desired behavior.

To conclude about the commutativity of the diagram it suffices to note that:
\[\Psi \circ \Xi \circ \Phi^*(U^q)=
(\theta^{q^{k-1}})^q=
\theta
=\Psi \circ \Xi \circ \Phi^*(V) \quad \mod  \F_q^*\]
 and similarly:
\[\Psi \circ \Xi \circ \Phi^*(V^q)=
\theta^q
=
\Psi \circ \Xi \circ \Phi^*(W)\quad \mod \F_q^*.\] 
Hopefully, since $\Psi \circ \Xi \circ \Phi^* \circ i$ is a morphism the equality in the finite field between images
of $A^qB-AB^q$ and $A(V,W)B-AB(V,W)$ holds too.

\paragraph{Extension of the diagram to bigger fields.}
As with classical Frobenius representation algorithm, the commutative
diagram can also be used when the coefficients of $A$ and $B$ are
taken in another extension $\F_{q^d}.$ In that case,  the commutative
diagram ends in the compositum of $\F_{q^d}$ and $\F_{q^k}$ (again
modulo $\F_{q}^*$). The only difference is that the identity in the
finite field is between the images of $A^qB-AB^q$ and
$A^{\pi}(V,W)B-AB^{\pi}(V,W)$, i.e. the coefficients of $A$ and $B$
need to be acted on by Frobenius.

\section{Harvesting Relations}
\label{collection}

\subsection{The usual systematic product}
As we can see in the previous diagram, our setting uses a mixture of 
the classical Function Field Sieve and of the Frobenius 
representation algorithms. As in the classical Function Field Sieve, our algorithm uses
function fields instead of polynomials when writing down
multiplicative relations. 
From Frobenius representation algorithms it
inherits the use of the systematic relation:
\begin{equation}
A^q\,B-A\,B^q  =\prod_{\alpha\in\PP_1(\F_q)}(A-\alpha\,B),
\end{equation}
where as in~\cite{DBLP:conf/asiacrypt/JouxP14}, when $\alpha$ 
is the point at infinity of $\PP_1(\F_q)$, the term
$A-\alpha\,B$ is used as a shorthand for $B$.
 For simplicity, we also use a bracket notation
and define:
$$
\bracket{A}{B}= A(V,W)\,B(U,V)-A(U,V)\,B(V,W).
$$
We underline that our bracket is $\F_q$-bilinear and
antisymmetric,  as in~\cite{DBLP:conf/asiacrypt/JouxP14}. 
Yet, we warn the reader of the difference between the definition
of our bracket and previous ones. Our bracket is equal to the entire fraction 
whereas the authors of~\cite{DBLP:conf/asiacrypt/JouxP14} 
only consider the numerator of this rational fraction.

\subsection{Choice of $A$ and $B$}
In the commutative diagram of Figure~\ref{fig:EllipticCD} and in the
above discussion, we indicate that relations are obtained from a
choice of two bivariate polynomials $A$ and $B$ in $U$ and
$V$. However, we need to specify how these polynomials are chosen and
which monomials they should contain.

As a preliminary, let us notice that $(A,B)$ and $(\alpha A, B)$ for
$\alpha \in \F_q$ lead to the same relation. Indeed, $(\alpha A)^q\,B-( \alpha A)\,B^q=
\alpha (A^q\,B-A\,B^q )  $ so the two divisors associated to the two
corresponding functions are equals. Thus $A$ and $B$ are chosen as
some kind of monic polynomials: the coefficient of the higher
monomial (in the lexicographic order for instance) must be equal to $1$.
Then we note that monomials divisible by $(UV)^2$
are not useful in $A$ and $B$. Indeed, when going to $\F_q[\C]$,
reduction modulo $S_3(U,V,x_1 )$ transforms these monomials 
into {\it smaller} monomials $U^2V, UV^2, U^2, UV, V^2, U, V$ and $1$.
Moreover,
it is natural to consider sets of monomials globally symmetric in $U$
and $V$. 

As a consequence for all those items, given a parameter $t\geq 1$ we construct
$A$ and $B$ as linear combinations of monomials from:
$$
\MM_t=\left\{ U^i, V^i, U^i\,V , U\,V^i | i \in [0\cdots t]\right\}.
$$
Each $\MM_t$ contains $4t$ distinct monomials\footnote{Not $4(t+1)$
  since each of $1$, $U$, $V$ and $UV$ are included twice.}.

\subsection{Defining a naive factor base}
\label{naive}
As shown in Figure~\ref{fig:EllipticCD}, we define  the factor base
as a subset of $Div^0(\E)$. We now explain how this subset is chosen.

\begin{defi}
\label{heightdefi}
Let $\DD$ be a divisor of $Div(\mathcal{V})$. The height of $\DD$ denotes the number of positive
points counted with multiplicity.
We write it $h(\DD)$.
\end{defi}

It matches with the following explicit definition: if the divisor $\DD$ is written as $\sum_{Q_i \in \mathcal{V}} e_i Q_i$ we have
$h(\DD) = \sum_{e_i>0} e_i$. 
Note that, for any prime or elementary divisor, the height is equal to
the degree of the corresponding place.

\begin{propo}
\label{sayuri}
If $f_1$ and $f_2$ are two functions of the function field associated to $\mathcal{V}$ then the following inequalities are verified:
\begin{enumerate}
\item $h(\Xi(f_1f_2)) \leq h(\Xi(f_1)) + h(\Xi(f_2))$.
\item $h(\Xi(f_1+f_2))\leq  \max(h(\Xi(f_1)) + \deg g_2\,,\,h(\Xi(f_2)) + \deg g_1) 
$
\\where $f_1$ (resp. $f_2$) has $g_1$ (resp. $g_2$) as denominator.
\item $h(\Xi(f_1)) \leq h(\Xi(f_2)) \quad \hbox{if $f_1$ and $f_2$ are polynomials such that $f_1$ divides $f_2$.}$
\end{enumerate}
\end{propo}
\begin{proof}
Let $f_i= e_i/g_i$ for $i=1,2$ be the two functions written such that $e_i$ and $g_i$ are two polynomials of the function field with no commun factor. Clearly, the function $f_1 f_2$ has~$e_1 e_2$ as numerator, thus
all the zeros of $f_1f_2$ are either a zero or $f_1$ or $f_2$.
  The inequality of $1.$ is strict if some simplifications appear (for instance if a zero of $f_1$ is also a pole of $f_2$).

Writing $f_1+f_2$ as $(e_1g_2+e_2g_1)/(g_1 g_2)$ 
we see that the number of zero of this sum is upper-bounded by the weighted degree of $e_1g_2$ or $e_2g_1$. Note again that some simplifications may appear.

Item $3.$ is straightforward.
\end{proof}

\paragraph{Left side.}Following the ideas of all Frobenius representation algorithms, we
define the factor base such that the images of $A - \alpha B$ in this set 
are small. Doing this improve the relation collection phase compared to a classical sieving.
Indeed, all elements in the left part of the diagram will belong to the factor base.
The relation collection phase produces divisors of the form $\Xi(\Phi^*(i(\prod(A-\alpha B))))$,
so, thanks to the fact that we consider morphisms, it yields a sum of divisors
$\sum (\Xi(\Phi^*(i(A-\alpha B))))$. As explained, we require all the divisors noted by:
\[\Xi(\Phi^*(i(A-\alpha B)))
\]
to be in the factor base.

Let us find the maximal height they can reach.
To do so we set $A$ and $B$ two linear combinations of  monomials in
$\MM_t$ 
and $t$~an integer parameter to define later. 
 All polynomials~$A - \alpha B$ are so linear combinations of monomials in
$\MM_t$ too.
In other words, we are considering divisors of functions of~$\F_q(\E)$ of the form $\Phi^*(\sum_{m \in \MM_t} a_m m) = \sum_{m \in \MM_t} a_m  \Phi^*(m) $ where $a_m$ are constants in the base field. From Proposition~\ref{sayuri} we know that
the height of the divisors we obtain in the left part of the diagram are dominated by the largest height
achieved for any $\Xi(\Phi^*(m))$, with $m \in \MM_t$.

Considering 
$h(\Xi(\Phi^*(U^tV^{t'}))) \leq t \, h(\Xi(\Phi^*(U)))+t' \,h(\Xi(\Phi^*(V)))$, we see that it suffices to determine both
the height of the divisors associated to the images of $U$ and $V$ in the function field of $\E$. The most significant monomials will be $U^tV$ and $UV^t$.

On the one hand we have $\Xi(\Phi^*(V)) =\Xi(X) = ([0,\sqrt{b}, 1]) + ([0,-\sqrt{b}, 1]) - 2 (\OO)$ 
where $\sqrt b$ is the element\footnote{Be careful, here we assume that the characteristic differs from $2$. If not, we just consider the corresponding degree-$2$ place.} in $\bar\F_q$ such that its square is equal to $b$. Hence:
\begin{equation}
\label{vincent}
h(\Xi(\Phi^*(V)))=2.
\end{equation}

On the other hand, $\Xi(\Phi^*(U)) =\Xi((Y+y_1)^2-(X-x_1)^3) - \Xi((X-x_1)^2)$. From 
$\Xi((Y+y_1)^2-(X-x_1)^3) = 2(-P_1)+(Q_1) + (Q_2) -4(\OO)$,
where $Q_1$ and $Q_2$ are two conjugated points of  a degree-2 place, and $ \Xi((X-x_1)^2) = 2(P_1)+2(-P_1)-4(\OO)$, it comes $\Xi(\Phi^*(U)) = (Q_1) + (Q_2) - 2 (P_1)$. We obtain:
 \begin{equation}
\label{ulysse}
h(\Xi(\Phi^*(U)))=2.
\end{equation}

Putting Equations~\eqref{vincent} and~\eqref{ulysse} together with the upper-bound, we conclude that the most significant monomials $U^tV$ and $UV^t$ have both height $2 t +2$. 
Yet $\Phi^*(U^tV)$ and $\Phi^*(UV^t)$ do not share the same denominator so to count their respective contribution in the height of divisors $\sum_{m \in \MM_t} a_m  \Phi^*(m)$ we need to add the contribution of the 
residual denominator. Namely, since $\Phi^*(U^t)$ brings the largest denominator, the height of the divisor of $\Phi^*(U^tV)$ does not change, but for the one of $\Phi^*(UV^t)$ we need to add the number of zeros corresponding to the denominator of $\Phi^*(U^t/U)$. We note that there is $2(t-1)$ such points.
To put it in a nutshell, the most significant monomial is $UV^t$ and \textbf{all divisors on the left are sum of divisors with a height upper-bounded by}~$\textbf{4t}=(2t+2)+2(t-1)$.

To conclude, 
starting the relation collection phase with $t=1$ it is natural to set the initial factor base as
included in the set of divisors of $Div^0$ with height equal or lower than~$4$. We emphasize that in this case,
all the divisors appearing in the left part belong to the factor base.

\begin{table}
\label{table}
\[
\begin{array}{ll}
\hbox{Functions of $\C$} &\hbox{Height of the associated divisors in $\E$}\\
\hline
1&0\\
U,\,V,\,W&2\\
UV, \, VW, \, UW&4\\
U+V, \, V+W, \, U+W&4\\
U^tV, \, UV^t &2t+2\\
U^tV+UV^t & \hbox{at most } 2t+2\\
U^tV^{t+1}W, \, UV^{t+1}W^t, \, U^tV^2W^t, \, U V^{2t}W& 4t+4\\
\end{array}
\]
\caption{Usual functions of $\F_q[\C]$ and their corresponding height in $\E$.}
\end{table}

\paragraph{Right side.}On the right part of the diagram, divisors are given through the extra variable $W$. We can compute the corresponding
height exactly as for $U$. Again it gives:
\begin{equation*}
h(\Xi(\Phi^*(W)))=2.
\end{equation*}
Let us consider the polynomials of $\F_q[U,V,W]$ given on this side and write this time
$\MM_t^{VW}=\left\{ V^i, W^i, V^i\,W , V\,W^i | i \in [0\cdots t]\right\}$. Sorting the monomials in the lexicographic order, we recall that the leading coefficient for both $A$ and $B$ can be chosen equal to $1$. 
 Thus, keeping the leading monomial $U^tV$ apart and calling $a_m$ (resp. $b_m$) the coefficients in $\F_q$
 of $A$ (resp. $B$) we obtain on the right side the polynomial:
\[
\begin{array}{l}
\bracket{A}{B}=A(V,W)B(U,V) -A(U,V)B(V,W)\\
\qquad \quad = \displaystyle \left(V^tW+\sum_{m \in \MM_t^{VW}  \setminus \{V^tW\}} a_m m\right)\left(U^tV+\sum_{m \in \MM_t \setminus \{U^tV\}} b_m m\right)\\
\qquad \qquad \qquad \qquad \qquad \qquad   \displaystyle  - \left(U^tV+\sum_{m \in \MM_t  \setminus \{U^tV\}}a_m m\right) \left(V^tW+\sum_{m \in\MM_t^{VW} \setminus \{V^tW\}} b_m m\right)
\end{array}
\]
Since the monomial $U^tV^{t+1}W$ vanishes it yields a linear combination of monomials 
where the three that dominate the height of the associated divisor are $UV^{t+1}W^t$, $U^tV^2W^t$ and $U V^{2t}W$. Indeed,
each variables $U$, $V$ and $W$ contributes the same way, thus, the most important monomials are those with the
highest additive degree. We note then that we have $h(\Xi(\Phi^*(UV^{t+1}W^t)))=h(\Xi(\Phi^*(U^tV^2W^t)))=h(\Xi(\Phi^*(U V^{2t}W)))=4t+4$. Yet, again, we need to carefully add
the zeros raised by the residual denominator.
 The contribution of $\Phi^*(U^tV^2W^t)$ is left unchanged but we must add the number of poles of $\Phi^*(W^t/W)$ to the height of the divisor associated to $\Phi^*(UV^{t+1}W^t)$ and the one corresponding to the denominator of $\Phi^*((UW)^t(UW)^{-1})$ to the height of the divisor associated to $\Phi^*(U V^{2t}W)$. From Proposition~\ref{sayuri} and since there are respectively $2(t-1)$ and $4(t-1)$ such points, we conclude that 
\textbf{all the divisors appearing on the right side are twice as large as factor base elements since they have height equal or lower than }$\textbf{8t}=4t+4+4(t-1)$. In particular, when $t=1$,
this gives divisors of height $8$ at most.

\paragraph{Complexity of the linear algebra with a naive factor base.}
A first and naive choice of factor base is made of all divisors of
height $\leq 4$.

However, this factor base is too large to be competitive when compared
to the best Frobenius representation algorithms. To show this, let us
briefly analyze the number of operations needed to
perform linear algebra with this factor base. The number of divisors in this naive
factor base is dominated by the number of degree~$4$ places on the
curve $\E$, corresponding to polynomials of degree~$4$ with
coefficients in $\F_q$. So the order of the factor base's size is dominated
by $q^4$.  Considering the Frobenius action of Section~\ref{sec frob}
that permits to divide the size of the factor base by $k\approx q$, we
obtain a factor base of size $O(q^3)$. Since there are $q$ terms in
each linear equations, performing a sparse linear algebra step can be
done in $O((q^3)^2 q)=O(q^7)$ operations.

As a comparison, we that the first phase of the algorithm
in~\cite{DBLP:conf/asiacrypt/JouxP14} only has a $O(q^6)$ complexity.
Thus, we need to improve our initial factor base.

\begin{table}
\label{table2}
\[
\begin{array}{ll}
\hbox{Functions of $\C$} &\hbox{Height of the associated divisors in $\E$}\\
\hline
A(U,V) - \alpha B(U,V) \quad\hbox{ where } \alpha \in \F_q&\hbox{at most } 4t\\
A(V,W)B(U,V) -A(U,V)B(V,W)&\hbox{at most } 8t
\end{array}
\]
\caption{Functions appearing on both side of the diagram, and their corresponding heights in $\E$. $A$ and $B$ are linear combinations of monomials from $\MM_t$.}
\end{table}

\subsection{Action of Frobenius and translation by $P_1$}
\label{sec frob}

Let $d$ be the largest possible height of an elementary divisor of
the factor base. We would like to explicit how the action of Frobenius on elements of the
finite field $\F_{q^k}$ is related to addition of $- P_1$ on the
elliptic curve $\E$. Considering the divisors:
\[\DD_Q= \sum_{i=0}^{d-1}(\Fr^i(Q))-d(\OO) \]
related to the place given by any point $Q \in\E$
and:
\[\DD_{Q-P_1}= \sum_{i=0}^{d-1}(\Fr^i(Q-P_1))-d(\OO)\]
 related to the translation of $Q$ by $-P_1$. We show that the two
 discrete logarithms satisfy a simple relation.  More precisely,
we have the following result:

\begin{lemma} Let $Q\in\E$ be a point with coordinates in $\F_{q^d}$.
We consider the divisors 
\[\DD_Q= \sum_{i=0}^{d-1}(\Fr^i(Q))-d(\OO)\quad \hbox{and}\quad \DD_{Q-P_1}= \sum_{i=0}^{d-1}(\Fr^i(Q-P_1))-d(\OO) \]
respectively related to the place given by the point $Q$ in $\E$ and the one given by the translation of $Q$ by $-P_1$.
Then:
\[
\Psi(\DD_{Q-P_1})=\Fr(\Psi(\DD_Q)) \cdot \textcolor{vertforet}{\Psi((-P_1)-(\OO))^{d N_d}}.
\]
 where $N_d$ is a common multiple of the cardinalities of $\E$
over each finite field $\F_{q^i} \subset \F_{q^d}$
\end{lemma}

\begin{proof}

Let us start from $\DD_Q$ the degree-$d$ divisor. We recall that to have a principal divisor
we need to consider 
$N_d\DD_Q$. Thanks to Miller algorithm we are able to recover a 
function~$f_Q$ with coefficients in the base field $\F_q$ such that $N_d\DD_Q$ is the divisor of this function.
By definition we obtain:
$\Psi(\DD_{Q})= f_Q(F-\OO)^{1/N_d}$. To simplify the notation let us write $\alpha = f_Q(\OO)^{1/N_d}$ that is
an element in $\F_q$. Hence on the one hand we have:
\begin{equation}
\label{alpha}
\begin{array}{lclr}
\Fr(\Psi(\DD_{Q}))&=& \Fr(f_Q(F)^{1/N_d})/\Fr(\alpha)&\\
&=&f_Q(\Fr(F))^{1/N_d}\cdot \alpha^{-1}& \hbox{since $\alpha \in \F_q$.}
\end{array}
\end{equation}
To link this expression to the divisor of $\Fr(Q)$, {\it i.e.} to the evaluation of $f_Q$
in the point $\Fr(F)=F+P_1$ we define the function $g_Q$
such that, for all $S$ in $\E$, $g_Q(S)=f_Q(S+P_1)$. Let us write the divisor of this new function.
Since a zero $S$ (resp. a pole) of $g_Q$ is such that $S+P_1$ is a zero of $ f_Q$ (resp. a pole),
we obtain:
\[
\begin{array}{lclr}
\Xi(g_Q)&=& N_d(\sum_{i=0}^{d-1}(\Fr^i(Q)-P_1)-d(-P_1))&\\
&=& N_d(\sum_{i=0}^{d-1}(\Fr^i(Q-P_1))-d(-P_1))& \hbox{thanks to the fact that $P_1$ lives in $\E/{\F_q}$.}\\
&=& \Xi(f_{Q-P_1})+N_d(d(\OO)-d(-P_1))& \\
&=& \Xi(f_{Q-P_1})- d N_d((-P_1) - (\OO))& \\
\end{array}
\]
Hence on the other hand we have:
\[
\begin{array}{lclr}
\Psi(\DD_{Q-P_1})&=& \Psi( \Xi(f_{Q-P_1}))&\\
&=& \Psi( \Xi(g_{Q}) + dN_d((-P_1) - (\OO)))& \\
&=& \Psi( \Xi(g_{Q}))\cdot\Psi((-P_1)-(\OO))^{d N_d}& \\
&=& g_Q(F)^{1/N_d} \alpha^{-1} \cdot (\alpha /g_Q(\OO)^{1/N_d}) \cdot \Psi((-P_1)-(\OO))^{d N_d}& \\
&\underset{\mod \F_q}{=}& g_Q(F)^{1/N_d}\alpha^{-1} \cdot \Psi((-P_1)-(\OO))^{d N_d}& \\
&\underset{\mod \F_q}{=}& f_Q(F+P_1)^{1/N_d}\alpha^{-1} \cdot \Psi((-P_1)-(\OO))^{d N_d}&\\
&\underset{\mod \F_q}{=}& \Fr(\Psi(\DD_{Q})) \cdot \Psi((-P_1)-(\OO))^{d N_d} \quad \hbox{from Equations \eqref{frobenius} and \eqref{alpha}.}\\
\end{array}
\]

\end{proof}

We emphasize that the green term is a constant term. Thanks to this action, we are able to reduce the size of the factor base 
by a factor~$k$ throughout the computations. Indeed,
if we know the discrete logarithm of $\Psi(\DD_Q)$
then we learn for free the discrete logarithms of~$\Psi(\DD_{Q-P_1}),
\Psi(\DD_{Q-P_2}), \cdots, \Psi(\DD_{Q -P_{k-1}})$.

\subsection{Getting a smaller factor base}

To be able to reduce the size of the initial factor base, and thus to decrease the complexity
of the linear algebra phase, we adapt the idea of systematic factors that was presented 
in~\cite{DBLP:conf/asiacrypt/JouxP14} to the elliptic case.
The idea was twofold: first extracting some
systematic factors that appear in every equation, second, restrict
the search to a sieving space that induce extra common factors. 
In this article, we choose to call these extra factors
{\it compelled factors} to underline the difference with previous ones.

\subsubsection*{Left part of the diagram: making $P_3$ a compelled point.}
In our case, our aim is to consider a subgroup of the sieving space
where $A$ and~$B$ are polynomials such that the associated divisors
always present a common (compelled) point. Here, choose to use the special
point $P_3 =3\,P_1$. As in~\cite{DBLP:conf/asiacrypt/JouxP14} we select three
generators $g_1$, $g_2$, $g_3$ in $\F_q[U,V]$ leading to divisors
going through $P_3$.  We propose to sieve on pairs of polynomials
$(A,B)$ such that $A = g_1 + \alpha g_3$ and
$B= g_1+\beta g_2 + \gamma g_3$ where
$\alpha, \beta, \gamma \in \F_q$.  Indeed, if $A$ (resp. $B$) is a
linear combination of those three generators and if $P_3$ is a zero of
$\Phi^*(i(g_j))$ for $j=1,2$~and~$3$ then it is also a zero of the
image of $A$ (resp. $B$) in $\E$. As a consequence and for the same
reason, $P_3$ is a zero of the image of $A-\alpha B$ too, where
$\alpha$ belongs to the base field.

\begin{lemma}
\label{Lyon}
Let $j$ be an integer in $[0,k-1]$ and assume that $P_0$ is a shorthand for $\OO$. Then:
\[
\begin{array}{rl}
&\Xi(\Phi^*(U-x_j)) = (P_{j+1}) + (-P_{j-1}) - 2(P_1), \\ 
 &\Xi(\Phi^*(V-x_j)) = (P_{j}) + (-P_{j}) - 2(\OO), \\
\hbox{and}  &\Xi(\Phi^*(W-x_j)) = (P_{j-1}) + (-P_{j+1}) - 2(-P_1).
\end{array}
\]
\end{lemma}
\begin{proof}
Let $j$ be an integer between $0$ and $k-1$. We recall that $x_j$ denotes the abscissa of~$P_j \in \E$ and $x_0=0$.
Then by definition:
\[\Phi(P_j)=[x_{j-1},x_{j},x_{j+1}]\]
 for all possible values of~$j$.
It means that, over the curve~$\C$, $\Phi(P_j)$ is a zero of
$U-x_{j-1}$, $V-x_j$ and $W-x_{j+1}$. Going back to the curve $\E$ it yields 
that the point $P_j$ is a zero
of $\Phi^*(U-x_{j-1})$, $\Phi^*(V-x_{j})$ and $\Phi^*(W-x_{j+1})$. 
Similarly, from: \[\Phi(-P_j)=[x_{j+1},x_{j},x_{j-1}]\] we get that $-P_j$ is a zero of 
$\Phi^*(U-x_{j+1})$, $\Phi^*(V-x_{j})$ and $\Phi^*(W-x_{j-1})$. 
Besides, writing $\Phi^*(U-x_{j}) $ as $ \left((Y+y_1)/(X-x_1)\right)^2-X-x_1 -x_j$ 
we see that $P_1$ is a pole of $\Phi^*(U-x_{j}) $ with multiplicity~$2$, and similarly $-P_1$
is a pole of $W-x_{j}$ with multiplicity~$2$. 
From $\Phi^*(V-x_{j})= X-x_j$ we conclude that $\OO$ is twice a pole too.
\end{proof}

Hence as generators we select:
\begin{equation}
\label{tiago}
\begin{array}{rl}
&g_1 = U-x_2, \\ &g_2 = V-x_3, \\ \hbox{and}  &g_3 = (U-x_2)(V-x_3).
\end{array}
\end{equation}
From Lemma~\ref{Lyon} we have $\Xi(\Phi^*(g_1))) = (P_3) + (-P_1) - 2(P_1)$ and $\Xi(\Phi^*(g_2))) = (P_3) + (-P_3) - 2(\OO)$.
Thus $\Xi(\Phi^*(g_3))) = 2(P_3) + (-P_1)+(-P_3) - 2(\OO) - 2(P_1)$. Clearly the point~$P_3$ is a positive point of each divisor.

Thus, if we start to sieve with a parameter $t$ equal to $1$, we obtain divisors on the 
left part of the diagram that have a height  lower or equal to $4$. Since $P_3$ is a positive point for all these
divisors, we are left with divisors that have a height lower or equal to $3$. We conclude with the definition of
the reduced factor base:
\[\FF = \left\{ d \in Div(\E) \quad  |\quad \hbox{$d$ is elementary and } h(d) \leq 3 \right\}.\]

\begin{remark}
The factor base is the same on both sides of the diagram.
\end{remark}

As previously, we can upper-bound the cardinality of the factor base by the number
of divisors with height lower than $3$, so the number of monic degree-$3$ polynomials in $\F_q$, which is
$q^3$.
Thanks to the Frobenius action, the base is reduces by a factor of
$k$, and the reduced factor base has size $O(q^3/k)$.
At the end, assuming that we get enough equations (we discuss this issue
 in Section~\ref{enough}), 
linear algebra recovers the discrete logarithms of the initial factor
base elements in $O((q^3/k)^2q)=O(q^7/k^2)$ operations. When $k$ is
close to $q$, this matches the result of~\cite{DBLP:conf/asiacrypt/JouxP14}.

\begin{remark}

Note that by contrast with polynomials, ideals of degree $4t$ are
determined from $4t$ monomials, instead of $4t+1$. This is unfortunate
because it forces us to increase the degrees to get enough degrees of
freedom when generating relations. However, this drawback is
counter-balanced by the reduction of the factor base size obtained by using the
action of Frobenius.

\end{remark}

\subsection{How to get enough relations with the reduced factor base}
\label{enough}
Thank to $P_3$ which is a compelled point we know that, on the left
part of the diagram, we directly have divisors that split into factor
base elements only. Now the questions are whereas we manage to easily
obtain small height divisors on the right part or not, and how many
relation we expect to write. We recall that, to be able to perform
linear algebra, we need as many relations as unknowns.

We have around $q^3/3$
unknowns and we sieve on $q^3$ pairs of polynomials $(A,B)$ consisting
in linear combinations of $g_1, g_2$ and $g_3$ as given in~\eqref{tiago}.
It means that we need a probability higher than $1/3$ to get a relation.
Since a degree-$d$ divisor is clearly linked with an irreducible polynomial
of the same degree as seen in~Section~\ref{preliminaries}, this probability
is assumed to be the same that a random polynomial of degree $d$
to factor into terms of degree at most~$3$.

For degree $d=8$, the probability is easy to compute. A polynomial
fail to factor into terms of degree at most~$3$ when one factor as
degree $4$ or more. Since there can only be a single factor when the
degree is $5$, $6$, $7$ or $8$ and at most two of degree $4$. Thus,
for large fields, the probability of success approaches
$1-(1/8+1/7+1/6+1/5+3/16+1/32)\approx 0.147.$
Unfortunately, this is much smaller than $1/3$.

\subsubsection*{Right part of the diagram: the two compelled points $P_2$ and $P_3$.}
Thus we need to look at the right part of the diagram more carefully. 
Going back to the analysis made in Section~\ref{naive}, \textit{Right part}
we see that choosing $A$ and $B$ as monic (in some sense) does not reduce 
the height of the associated divisor. Hence, for $(A,B)$ a pair of linear combinations
of $g_1, g_2$ and $g_3$ as in~\eqref{tiago}, the divisor:
\[\Xi(\Phi^*(\bracket{A}{B}))\]
 has still a height of $8$. Because $P_3$
is a zero of $\Phi^*(A(U,V))$ and $\Phi^*(B(U,V))$,
we note that $P_3$ is a zero of $\Phi^*(\bracket{A}{B})$ too.
Yet it is not enough and we need to extract another compelled positive point
of the associated divisor to the image of $A(V,W)$ (resp. $B(V,W)$) over $\E$. 
We start by underlining that $g_1, g_2$ and $g_3$ respectively becomes
$V-x_2, W-x_3$ and $(V-x_2)(W-x_3)$, when sending $V$ to $W$ and $U$ to $V$. Thus according to Lemma~\ref{Lyon}, the point 
$P_2=2\,P_2$ is a zero of all
the generators, and so a zero of $\bracket{A}{B}$ as $P_3$.
We conclude that we are left with a divisor of height at most~$6$. The probability that it
splits into a sum a divisors with height at most~$3$ is so roughly equals to:
\[1-(1/6+1/5+1/4) \approx 0.383 >1/3,
\]
as $q$ grows.
As a consequence, we heuristically expect to get a linear system with enough 
equations to get the discrete logarithms of all elements of $\FF$ in $O(q^5)$ operations.

 \noindent \fcolorbox{bleucecile}{white}{\parbox{\linewidth \fboxrule \fboxsep}{
\begin{heuristic}
The heuristic in this linear algebra step and in all the following ones comes from the
fact that we have no argument to prove we really get enough equations. We can count
them and expect that when their number slightly exceeds the number of unknows,
we are able to find a solution. Yet,
nothing provably indicates whether the kernel of our matrix of relations has dimension~$1$ or not.
\end{heuristic}}}

\section{Extended Factor Base and Individual Discrete Logarithm}
\label{extended}
We only sketch here the last two main steps of our practical
algorithm, the computation of an Extended Factor Base and the
Individual Discrete Logarithm step.  Indeed, they are
 an adaptation of the techniques that already exist
 for Frobenius representation algorithms to our setting.
 
\subsection{From divisors of height $3$ to divisors of height $4$}
As done in Frobenius representation algorithms, we extend the initial
factor base and now include all elementary divisors up to height
$4$. Thanks to the Frobenius action, there are approximately
$q^4/4k\approx q^3/4$
unknowns.  The naive approach we showed earlier gives the
desired logarithms at a cost of $O(q^7)$ arithmetic operations (or
$O(q^9/k^2)$ when $k$ is away from $q$).

\subsubsection*{Practical speed up with regrouping}
To speed up the computation of height $4$ divisors, it is possible to
decompose the height~$4$ factor base into small groups, in a way
similar to~\cite{DBLP:conf/asiacrypt/JouxP14}, in order to perform
several linear algebra steps on these small groups, instead of a
single big linear algebra step. In Appendix~\ref{bunches}, we give
details on how to produce relation in these groups. One technicality
is the interaction of the groupings with the reduction of the factor
base size given the action of Frobenius.  

Once the height $4$ divisors are obtained, it is a simple matter to
continue extending the factor base to height $5$ divisors. For that
final step, no additional linear algebra is needed. It suffices to
keep relations where a single height $5$ divisors appears, the rest
being of lower degree. See Appendix~\ref{bunches} for a detailed
explanation.

However, for the height $4$ extension, the expected number of produced
relations seems to be slightly too low asymptotically to guarantee its
success. Nevertheless, we tested the method on a practical example to
check its viability. Namely, for the target finite field
$\F_{3^{1345}}$, we were able to compute logarithms for an extended
factor base comprising divisors of height up to $5$.  This was done by
choosing a curve of (prime) cardinality $269$ over $\F_{243}$.
Studying the exact behavior of the height $4$ extension to understand
this gap is thus a matter of further research.

\subsection{Computing Individual Discrete Logarithms}
To solve the discrete logarithm problem in our target finite field, we
need not only to know the logarithms of extended factor base elements
but to be able to compute the discrete logarithm of any arbitrary
element. This is the aim of this paragraph. Various descent phases
were previously proposed by various authors, the idea is to show how
to adapt to our context. In practice, one can use the bilinear
descent, the classical descent and the zig-zag descent or a mix of
them. Indeed the quasi-polynomial descent
of~\cite{DBLP:conf/eurocrypt/BarbulescuGJT14} is unlikely to be
practical for currently accessible computations.

For the classical descent which simply consist in writing the target
finite field element whose discrete logarithm is wanted as a product
or quotient of relatively low-degree polynomials in $\theta$, no
adaptation is needed. We only need to check that any polynomial in
$f(\theta)$ can be injected in the commutative diagram. This is simply
done by written the divisor of $f(V)$ since $\Psi(V)=\theta$. When $f$
is irreducible, the corresponding divisor is either the sum of two
elementary divisors of height $\deg(f)$ or a single elementary divisor
of height $2\deg(f)$.

We illustrate the adaptation with the bilinear and zigzag descents:

\paragraph{Bilinear descent for our setting.}
The bilinear descent step is easy to adapt
from~\cite{DBLP:conf/asiacrypt/JouxP14}.  Remark that we usually need
to unbalance the degrees of freedom in $A$ and $B$, thus
choosing different sets of generating polynomials. Instead of
constructing the polynomials just from $1$, $U$, $V$ and $UV$ we built
them from higher degree polynomials in $\MM_{t_a}$ and
$\MM_{t_b}$ respectively. We assume that $t_a\geq t_b$.  Let us first
analyze the case where we use all these monomials, remembering that
there are $4t_a$ and $4t_b$ of them. As usual, we force
$A$ and $B$ to be monic and remove the head monomial of $B$ from
$A$. All of the other coefficients are replaced by a corresponding
formal unknown. Thus, the polynomial $A$ contains $4t_a-2$ unknowns. If
$t_a\neq t_b$, $B$ contains $4t_b-1$ monomials. If $t_a = t_b$, we can remove an
extra unknown from $B$. Furthermore, we know that the height of
factors of the form $A-\alpha B$ is upper bounded by $4t_a$.
We also know from Table~\ref{table2} that the
height of the bracket is at most $8t_a$

If we want to adjust the values modulo $4$ of the number of degrees of
freedom, it is necessary to use compelled points. More precisely, we
can force $A$ and $B$ to go through one, two or three forced
points. This reduces the degrees of freedom by the same amount on both
sides. It also reduces height on the left by the same value and
heights on the right by its double.

As in Frobenius representation algorithms, the coefficients of each
monomials in $\bracket{A}{B}$ are bilinear (or linear or constant) in the $A$ and $B$ unknowns.
To force an elementary divisor of degree $d$ to appear in $\bracket{A}{B}$, it suffices to
require that the bracket vanishes when evaluated at each of the $d$
conjugate points corresponding to the associated prime divisor.
This yields a bilinear system of $d$ equations in the $A$ and $B$
unknowns. This equation can be solved using Gröbner basis techniques
exactly as in the case of  Frobenius representation.

The only extra (and minor) restriction here is the relation between the number of
$A$ variables and $B$ variables modulo $4$.

\subsection{Zig-zag descent.}
The zig-zag descent seems to be the best option to achieve
provable quasi-polynomial complexity. In particular, it is used both
in ~\cite{KW19} and~\cite{GuidoLido2}. 
As a consequence, we also show how to adapt it to our
setting. As it is more lenghty to describe than the bilinear descent,
we assign a separate section to it.

\subsubsection*{Short recap on the zig-zag descent.}
\label{zigzag}
First the main idea is to adapt the zig-zag descent presented in~\cite{DBLP:journals/iacr/GrangerKZ14a} to our setting. Let us give an insight of this descent in the classical settings. We call $z$ our target, which is an irreducible polynomial  in $\F_q[X]$ of degree \footnote{Indeed, one can use  Wan's theorem~\cite[Theo 5.1]{Wan97generatorsand} to ensure that any field element is equivalent to an irreducible polynomial of degree a power of~$2$ only slightly larger than the extension degree~$k$.} $2d=2^{t+1}$.
One crucial point of this method is that for any relation in $\F_{q^2}[X]$ implying
degree-$d$ polynomials, one can find a relation in the subfield  $\F_q[X]$ at the price of having polynomials of
degree twice as large. Thus,  in order to make $z$ appear in a polynomial relation of $\F_q[X]$, we write it as a product of two degree-$d$ conjugated polynomials $\tilde z$ and ${\tilde z}^*$ 
over the extended field~$\F_{q^2}[X]$ and we try to get one relation (in the extended field) involving one of this degree-$d$ polynomials.  Multiplying by the 
same conjugated relation we would obtain a relation (in the subfield) where $z$ appears.

Recursively manipulating this trick on a tower of extensions as presented 
in Figure~\ref{toureiffel}, we write in fact $z$ as a product of conjugated degree-$2$ 
polynomials over $\F_{q^{2^t}}[X]$. 
 Indeed, this descent method rests upon the existence of an extended 
field in which any degree-$2$ polynomials evaluated in $\theta$ can be written 
as product of linear evaluations in $\theta$. Thus at the end, 
we get a relation of the form $z(\theta) = \prod_i L_i(\theta)$ where $L_i$ are linear polynomials.

\begin{figure}[h]
\centering
\begin{tikzpicture}[commutative diagrams/every diagram]
  \node (P0) at (0cm,0cm) {$\begin{array}{c}
\F_{q^{2^t}} \left[ X\right] \\
\prod z_i, \hbox{ where $z_i$ are degree-$2$ polynomials}
\end{array}$};
  \node (P1) at (0cm,-2cm){$\begin{array}{c}
\F_{q^{2^{t-1}}} \left[ X\right]\\
\prod z_i, \hbox{ where $z_i$ are degree-$4$ polynomials}
\end{array}$};
  \node (P2) at (0cm,-3.5cm){$\vdots$};
  \node (P3) at (0cm,-5cm){$\begin{array}{c}
\F_{q^{4}}  \left[ X\right] \\
\prod z_i, \hbox{ where $z_i$ are degree-$2^{t-1}$ polynomials}
\end{array}$};
  \node (P4) at (0cm,-7cm){$\begin{array}{c}
\F_{q^{2}} \left[ X\right] \\
 \tilde z \cdot \tilde z^*, \hbox{ where $\tilde z$ and $\tilde z^*$ are degree-$2^{t}$ polynomials}
\end{array}$};
  \node (P5) at (0cm,-9cm){$\begin{array}{c}
\F_{q}  \left[ X\right]  \\
z \hbox{ a degree-$2^{t+1}$ polynomial}
\end{array}$};
  \path[commutative diagrams/.cd, every arrow, every label]
    (P1) edge node {} (P0)
    (P2) edge node {} (P1)
    (P3) edge node {} (P2)
    (P4) edge node {} (P3)
    (P5) edge node {} (P4);

\end{tikzpicture}
\caption{\label{toureiffel} Tower of extensions over the base field $\F_q$ in the classical zig-zag descent.}
\end{figure}

To adapt this descent to our settings, a idealized method would be  to 
exhibit a sufficiently large extension of the curve $\C$ in which any
height-$2$ divisor can be written not as a degree-$2$ place but as a
sum of points on this exact extension (and not the larger following
one). This precisely would have translated the requirement that all
degree-$2$ polynomials split in linear polynomials when the extension
degree of the field is sufficiently large. Unfortunately, to the best
of our knowledge,  this ideal adaptation isn't possible.

On the technical side, we see that the method is much easier to
describe when computing logarithm in $\F_{q^k}$ for an odd extension
degree $k$. Indeed, in that case, the compositum of $\F_{q^k}$ and any
extension $\F_{q^{2^i}}$ is simply $\F_{q^{k2^i}}$. Making this
assumption is very convenient to describe the adaptation to the
elliptic representation.

\subsubsection{Elliptic zig-zag descent.}
We now go back to the elliptic representation setting, with the
additional restriction that the extension degree $k$ is odd.

\paragraph{Points and divisors over extensions.}
As mentioned in the last paragraph of Section~\ref{representation},
the commutative diagram in Figure~\ref{fig:EllipticCD} can be used not
only over $\F_{q}$ but also over extensions. We now give more details
for $\F_{q^d}$, assuming that $d$ and $k$ are coprime.

This we now use polynomials $A$ and $B$ with coefficients
in the larger field $\F_{q^d}$. Everything remains almost identical, except the
definition and properties of the bracket. With a larger field, we use:
$$
\bracket{A}{B}_*= A^{\pi}(V,W)\,B(U,V)-A(U,V)\,B^{\pi}(V,W),
$$
where $A^{\pi}$ denotes the polynomial derived from $A$ by raising each
coefficient of $A$  to the power $q$  (while keeping the same
monomials).
This new bracket $ \bracket{\cdot}{\cdot}_*$ is $\F_q$-bilinear (but
not $\F_{q^d}$-bilinear).

\paragraph{Bootstrapping the descent.}
\label{chemin}
Let $z\in \F_{q^k}$ be our target arbitrary element for which we want to find a discrete logarithm.
Thanks to the diagram of Figure~\ref{fig:EllipticCD}, we know that there exists a polynomial $\Pol$ in $\F_q[U,V]$ such that:
\[z=\Psi(\Xi(\Phi^*(\Pol(U,V)))).
\]
In fact, there are many such polynomials. We choose $\ell$ such that
$2^\ell>k$ and search a representation by a polynomial $\Pol$  in $\F_q[U,V]$ such that:
 \begin{enumerate}
\item $z=\Psi(\Xi(\Phi^*(\Pol(U,V)))).$
\item $h(\Xi(\Phi^*(\Pol(U,V)))) = 2^\ell.$
\item $\Xi(\Phi^*(\Pol(U,V)))$ exactly corresponds to a place of
  degree $2^\ell.$
\end{enumerate}
 
 Let us call $p_z$ such a place in  $\Sigma_{\F_{q}(\C)}$. We could
 lift it to   $\Sigma_{\F_{q^{2^\ell}}(\C)}$ so that it corresponds to
 $2^\ell$~points. However, for the rest of the method, it suffices to
 decompose it into degree-$8$ places. Theses places appear in $\Sigma_{\F_{q^{2^{\ell-3}}}(\C)}$.

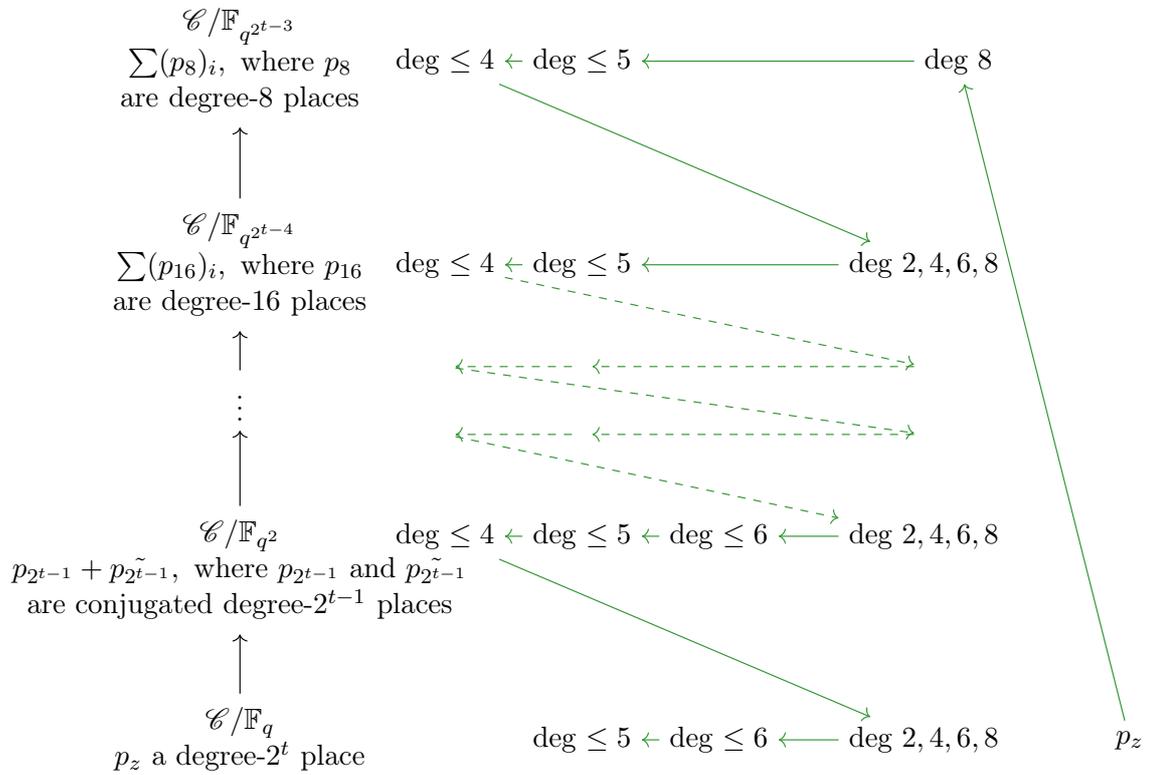
\begin{figure}[p]
\centering
\begin{tikzpicture}[commutative diagrams/every diagram,scale=0.9]
  \node (P1bis) at (0cm,-6cm){$\begin{array}{c}
\C/\F_{q^{2^{t-3}}}\\
\sum (p_8)_i, 
\hbox{ where $p_8$}
\\ \hbox{are degree-$8$ places}
\end{array}$};
  \node (P2) at (0cm,-9cm){$\begin{array}{c}
\C/\F_{q^{2^{t-4}}}\\
\sum (p_{16})_i, \hbox{ where $p_{16}$}\\
 \hbox{are degree-$16$ places}
\end{array}$};
  \node (P3) at (0cm,-11cm){$\vdots$};
  \node (P4) at (0cm,-13.5cm){$\begin{array}{c}
\C/\F_{q^{2}} \\
p_{2^{t-1}} + \tilde{p_{2^{t-1}}}, \hbox{ where $p_{2^{t-1}}$ and  $\tilde{p_{2^{t-1}}}$}
\\ \hbox{are conjugated degree-$2^{t-1}$ places}
\end{array}$};
  \node (P5) at (0cm,-16cm){$\begin{array}{c}
\C/\F_{q}  \\
p_z \hbox{ a degree-$2^{t}$ place}
\end{array}$};
\node (Q1) at (13cm,-16cm){$p_z$};
\node (Q2) at (10.5cm,-6cm){deg $8$};
\node (Q4) at (5cm,-6cm){deg $\leq 5$};
\node (Q5) at (3cm,-6cm){deg $\leq 4$};
\node (Q6) at (10cm,-9cm){deg $2,4,6,8$};
\node (Q8) at (5cm,-9cm){deg $\leq 5$};
\node (Q9) at (3cm,-9cm){deg $\leq 4$};
\node (Q10a) at (10cm, -10.5cm){};
\node (Q10b) at (7cm, -10.5cm){};
\node (Q10c) at (5cm, -10.5cm){};
\node (Q10d) at (3cm, -10.5cm){};
\node (Q10e) at (10cm, -11.5cm){};
\node (Q10f) at (7cm, -11.5cm){};
\node (Q10g) at (5cm, -11.5cm){};
\node (Q10h) at (3cm, -11.5cm){};
\node (Q11) at (10cm,-13cm){deg $2,4,6,8$};
\node (Q12) at (7cm,-13cm){deg $\leq 6$};
\node (Q13) at (5cm,-13cm){deg $\leq 5$};
\node (Q14) at (3cm,-13cm){deg $\leq 4$};
\node (Q15) at (10cm,-16cm){deg $2,4,6,8$};
\node (Q16) at (7cm,-16cm){deg $\leq 6$};
\node (Q17) at (5cm,-16cm){deg $\leq 5$};

  \path[commutative diagrams/.cd, every arrow, every label]
(P2) edge node {} (P1bis)
    (P3) edge node {} (P2)
    (P4) edge node {} (P3)
    (P5) edge node {} (P4);
  \path[commutative diagrams/.cd, every arrow, every label]
    [color=vertforet]
   (Q1) edge node {} (Q2)
    (Q2) edge node {} (Q4) 
    (Q4) edge node {} (Q5)
    (Q5) edge node {} (Q6)   
     (Q6) edge node {} (Q8) 
    (Q8) edge node {} (Q9)
(Q11) edge node {} (Q12) 
(Q12) edge node {} (Q13) 
(Q13) edge node {} (Q14) 
    (Q14) edge node {} (Q15) 
    (Q15) edge node {} (Q16) 
    (Q16) edge node {} (Q17) 
    ;
  \path[commutative diagrams/.cd, every arrow, every label]
  [color=vertforet,dashed](Q9) edge node {} (Q10a)
  (Q10a) edge node {} (Q10c) 
  (Q10c) edge node {} (Q10d)
 (Q10d) edge node {} (Q10e) 
 (Q10e) edge node {} (Q10g) 
  (Q10g) edge node {} (Q10h)
  (Q10h) edge node {} (Q11) ;
  
\end{tikzpicture}
\caption{\label{pearltower} Tower of elliptic curves extensions in the elliptic zig-zag descent. The path in green represents how we decompose $p_z$ in smaller degree places over higher extensions during the algorithm.}
\end{figure}

\paragraph{Descending degree-$8$ places.}
Using a series of relations based on the bracket $\bracket{A}{B}_*$,
there is a way to express the logarithm of the divisor corresponding
to a degree-$8$ place as a sum of logarithms of divisors of degree at
most $4$. Once this is done, we can pair conjugates divisors and go
one step down in the tower of quadratic extension. This at most doubles
the height of divisors. Iterating the process, we now encounter places
of degree $6$ and $8$ whose divisors need to be expressed as
combination of divisors of degree at most $4$. Finally, at the bottom
of the tower, everything can be expressed using divisors of height at
most $4$, this in turn permit to compute the logarithm of $z$.

 Keeping this strategy in mind we now describe the transformation of
 logarithm of divisors into sums of divisors of lower height. More
 precisely, we first transform degree $8$ places as sums using divisors of
 height at most $6$. Places of degree $6$ can be expressed using
 divisors of height up to $5$. Finally place of degree $5$ are
 transformed using divisors of height up to $4$.

 The exact degrees appearing in the descent strategy depend on the
 relative position in the tower of extension. Except at the lower
 levels, it is possible to descent directly from degree $8$ to degree
 $5$ and from degree $6$ to degree $4$. Except at the lower levels, it
 even possible to descend from degree $4$ to degree $3$.

 Thus, from a practical point of view,there are two essentially
 equivalent options for the descent. Either one starts from a
 degree-$8$ place and encounters descent steps from $8$ to $5$ then
 $4$ and descent steps from $6$ to $4$, except in the lower levels
 where longer chains from $8$ to $6$ then $5$ and finally $4$
 appear. Or one starts from a degree~$4$ place and encounters descent
 steps from $6$ to $4$ then $3$ and steps from $4$ to $3$. At the
 lower levels, this approach gets stuck.

In our presentation, we choose the approach that starts from a
degree-$8$ place. Note than in the context of provable algorithms,
using degree~$8$ possibly leads to a more difficult proof.

 \paragraph{Degree-$5$ places.}
We start with degree $5$ places since it is slightly simpler and
illustrates the general idea. We let $d$ be the power of $2$
corresponding to our current position in the tower of extensions.

Again, we create somehow relations from:
\[\prod_{\alpha\in\PP_1(\F_q)}(A-\alpha\,B) =\bracket{A}{B}_*  ,
\] where $(A,B)$ is a pair of polynomials with coefficients in
$\F_{q^d}$. For degree $5$, the polynomials are built from
monomials in $\MM_1$, i.e. from $1$, $U$, $V$ and $UV$.
To check whether there are enough degrees of freedom to force a place
too appear, we need to consider how many (non equivalent) pairs of
candidates relations we can try. 

Since we use the new bracket instead of ~$\bracket{\cdot}{\cdot}$, the
counting changes slightly. Before considering the property of the
bracket, there is a total of $8$ coefficients in $\F_{q^d}$, four in
each of $A$ and $B$.  Remark that, for any $\Lambda \in \F_{q^d}$, we
have
$\bracket{\Lambda A}{\Lambda B}_* = \Lambda^{q+1} \bracket{A}{
  B}_*$. Simultaneously, the left side corresponding to
$(\Lambda A, \Lambda B)$ is
$ \Lambda^{q+1} \prod_{\alpha \in \PP_1(\F_q)} (A-\alpha B)$.  Since
$ \Lambda^{q+1} $ appears on both sides, we see that $(A,B)$ and
$(\Lambda A, \Lambda B)$ generate the same equation.  Thus, we can set
the leading coefficient of $A$ to $1$.  This removes one of the
coefficients.

In addition, because of the $\F_q$-linearity of the bracket, for any
$\lambda \in \F_{q}$, we have
$\bracket{A}{B-\lambda A}_* = \bracket{A}{B}_*.$ Thus we can fix one
component of the leading coefficient of $B$ to $0$. Then, using
$\bracket{A}{\lambda B}_* = \lambda \bracket{A}{B}_*$ we can fix one
component of another coefficient of $B$ to $1$. Finally, thanks to the
relation $\bracket{A-\lambda B}{B}_* = \bracket{A}{B}_*,$ we can set
the corresponding component in $A$ to $0$.

This decreases the numbers of degrees of freedom to $7-3/d>5,$ when
$d>1$. Thus, we have enough degrees of freedom available. In the case
$d=1$, we are in the base field where the logarithms of the degree-5
places have already been precomputed as part of the extended factor
base.

Let $p_5$ be a place of degree $5$ in $\Sigma_{\F_{q^d}(\C)}$. Using
a variation on bilinear descent and solving a bilinear system of
equations in $6$ unknowns over the extension field $\F_{q^d}$, we can
obtained a relation involving $p_5$. Since the number of variables is
a small constant, this is a very efficient computation.

The relation can be written in the form:
\begin{equation*}
 \sum_{D \in Div \,| \, h(d) \leq 4} D = (p_5) + D_3
\end{equation*} where $D_3$ is a divisor of height $3$. This shows that we can descend 
any divisors of height $5$ as a sum of divisors of height at most~$4$. Note that $D_3$ and the divisors on the left may not be
elementary, however, in that case they decompose into elementary
divisors of lower degrees.

Note that we do not prove here that such a decomposition always
exists. Instead out counting of the degrees of freedom gives heuristic
support to this fact. It might be possible to adapt the proofs 
of~\cite{GuidoLido2} or~\cite{KW19} to our specific setting.

\paragraph{Degree-$6$ places.}
For degree $6$, there are two options depending of the extension
degree~$d$.

When $d\geq 4$, we can again build relations using only the
monomials $1$, $U$, $V$ and $UV$. In this case, it gives $7-3/d>6$
degrees of freedom. Thus, we can directly descend  to a sum of
divisors of height at most~$4$.

For the remaining cases, $d=2$ or $d=1$, we need to use monomials
from $\MM_2$ to provide more degrees of freedom. However, if we use
them all, the height of the left-hand factor become $8$ and the height
of the bracket is $16$. To control this explosion, it suffices to fix
three (essentially arbitrary) compelled points and keep a basis of all
functions going through these $3$ points. This basis contains $5$
polynomials, say $g_1$, \dots, $g_5$. Forming $A$ and $B$ as linear
combinations of the $g_i$s induces a systematic factor of total height
$3$ in every term $A-\alpha B$ (corresponding to the compelled
points). Furthermore, this systematic factor also appears in the
decomposition of the bracket together with an extra systematic factor
also of height $3$. This extra factor corresponds to the compelled
points translated by $-P_1$. Thanks to the systematic factors, the
height of the left becomes $5$ while the height of the right becomes
$10$. There is a total of $10$ coefficients in $A$ and $B$, which
corresponds to $9-3/d\geq 6$ degrees of freedom when removing
identical relations as in the previous case. More precisely, we can
fix the coefficient of $g_1$ in $A$ to $1$, one component of the
coefficient of $g_1$ in $B$ to $0$, and one component of $g_2$ to $1$
in $B$ and $0$ in $A$.

Solving a bilinear system, we can find coefficients that lead to an
equation:
\begin{equation*}
 \sum_{D \in Div \,| \, h(d) \leq 5} D = (p_6) + D_4
\end{equation*} where $D_4$ is a divisor of height $4$. This expresses
the logarithm $p_6$ as a sum of logarithms of divisors of height at
most $5$.

\paragraph{Degree-$8$ places.}
For degree $8$ places, we proceed as in the second method for degree
$6$. We use monomials from $\MM_2$. With three compelled points as in
degree $6$, we have $9-3/d$ degree of freedom. This is more than $8$
as soon as $d\geq 4$. In this case, we can write the logarithm $p_8$
as a sum of logarithms of divisors of height at most $5$.

When $d$ is $1$ or $2$, we use only two compelled points. We thus have
a basis of $6$ polynomials and $9-3/d\geq 8$ degrees of freedom. The
height after removing the systematic factors become $6$ for the left
factors and $12$ for the bracket. Thus, in the lower levels of the
tower of extension, we can write the logarithm $p_8$ as a sum of
logarithms of divisors of height at most $6$.

\paragraph{Practical (un)efficiency of the approach.}
In the Frobenius representation zig-zag, every step down the tower was
based on the creation of one relation. As a consequence, at every
level, the total number of elements under consideration was multiplied
by $O(q)$.

By contrast, here, we need two levels of relations for each of the
middle levels of the tower. As a consequence, the total number at each level is
multiplied by $O(q^2)$, which makes this approach much less appealing in
practice.

\ifanonymous
\else
\section*{Acknowledgments}
This work has been supported by the European Union's H2020
Programme under grant agreement number ERC-669891.
\fi

\bibliographystyle{alpha}
\bibliography{biblio}

\newpage
\appendix

\section{Details on the curve model}
\label{appendix}
In order to understand our model of $\C$, we analyze how points of
$\E$ are mapped by $\Phi$  to the model $\C$ on the three variables
$U$, $V$ and~$W$.  We assume for simplicity that $\E$ is given by a
reduced Weirstrass equation, but this can be generalized to include
characteristic $2$ and $3$.

\paragraph{Equations of $\C$.}
We assume that $\E$ is given by a reduced Weierstrass equation:
$$
\E: Y^2=X^3+a\,X+b.
$$
In this case, the third summation polynomial is given by:
$$
S_3(X_1,X_2,X_3)=4\sigma_1(\sigma_3+b)-(\sigma_2-a)^2,
$$
where the $\sigma_i$ are the symmetric polynomials:
\begin{eqnarray*}
\sigma_1& =& X_1+X_2+X_3,\\
\sigma_2& =& X_1\,X_2+X_1\,X_3+X_2\,X_3\quad \mbox{and}\\
\sigma_3& =& X_1\,X_2\,X_3.
\end{eqnarray*}

Let us first consider the variety given by the equations $S_3(U,V,x_1)=0,$
$S_3(V,W,x_1)=0$ and $S_3(U,W,x_2)=0.$ To determine its components, 
let us consider its intersection with the hyperplane $U=W$. This
intersection is described by $S_3(U,V,x_1)=0$, $U=W$ and
$S_3(U,U,x_2)=0.$
From the third equation that is a degree-$4$ polynomial in $U$, we know that $U$ has
finitely many values. Thus we want to remove the extraneous points
lying in this hyperplane.

We look so at the components in the complement of this hyperplane and assume that
$U\neq W$. 
In this case, since
$S_3(U,V,x_1)-S_3(V,W,x_1)$ is divisible by $U-W$ we obtain a lower degree polynomial, namely:
$$
S_\delta=\frac{S_3(U,V,x_1)-S_3(V,W,x_1)}{U-W}.
$$
The variety defined thanks to the equations
$S_3(U,V,x_1)=0$, $S_\delta=0$ and $S_3(U,W,x_2)=0$ is now irreducible. We call it $\C$
and prove that it is a genus $1$ curve isomorphic to $\E$. To see that, let us give rational maps between
$\E$ and $\C$.

\paragraph{Mappings between $\E$ and $\C$.}
In the forward direction, let us consider the rational map:
\[
\begin{array}{cccl}
\Phi : & \E & \rightarrow & \C \\
 &Q & \mapsto & (x_{Q-P_1},x_Q,x_{Q+P_1})
\end{array}
\]
Every
point $P$ in $\E$ is such that $\Phi(P) \in \C$.

Besides, the images of the three points $\OO$, $P_1$ and $-P_1$ are
at infinity on $\C$ and that, by homogenization, we may check that
there are exactly three points at infinity on $\C$. 

As usual, $\Phi$ induces a map $\Phi^{*}$ from the function field
$\F_q(\C)$ to $\F_q(\E)$ (expressed with the two variables $X$ and
$Y$) using the following replacement:
\[
\begin{array}{cccl}
  \Phi^*: &  \F_q(\C)& \mapsto & \F_q(\E) \\
  &U& \mapsto &  \left(\frac{Y+y_1}{X-x_1}\right)^2-X-x_1,  \\
  & V  &\mapsto &   X, \\
  &  W &\mapsto &   \left(\frac{Y-y_1}{X-x_1}\right)^2-X-x_1.
\end{array}
\]
where $y_1$ is the ordinate of the point $P_1$ in $\E$. 
Developing and using the curve equation, the images of $U$ and $W$ can
be respectively simplified to:
\begin{eqnarray*}
  U & \mapsto &
  \frac{x_1X^2+(a+{x_1}^2)\,X+a\,x_1+2\,b+2\,y_1\,Y}{(X-x_1)^2}, \\
  W & \mapsto & 
  \frac{x_1X^2+(a+{x_1}^2)\,X+a\,x_1+2\,b-2\,y_1\,Y}{(X-x_1)^2}.
\end{eqnarray*}
In this form, it is clear that the map can be easily inverted when
$X\neq x_1$. Moreover, given a pair $(U,V)$ values we can compute
the value of $W$ and similarly, from $(V,W)$ we can compute $U$.

When $X=x_1$, we have two possibilities $\Phi(P_1)=(\infty,x_{P_1},x_{P_2})$ and
$\Phi(-P_1)=(x_{P_2},x_{P_1},\infty)$. These are distinct (unless $P_1$ has order $2$), which
means that $\Phi$ is a bijection and thus an isomorphism.

\section{Analysis of the invertibility of $N_D$}
\label{rightcurve}
\label{carefulconstruction}
In this appendix, we analyze the condition that appears in
Section~\ref{diagram} when explicitly writing down the definition of
the morphism $\Psi$. Indeed, as previously explained, we need $N_D$ to
be invertible in the group where we want to compute discrete
logarithms.

This analysis requires us to follow standard practice and first
decompose our target group~$\F_{q^{k}}^*$, in order to apply
Pohlig-Hellman algorithm~\cite{DBLP:journals/tit/PohligH78}.  Thanks
to this method it suffices to compute discrete logarithms in all prime
order subgroups of $\F_{q^{k}}^*$.

An important technicality is that we would need to
first factor $q^k-1$. Unfortunately, this would completely dominate
the cost of computation. However, to study the invertibility of $N_D,$
we do not need to factor $q^k-1$ fully, the existence of the
factorization suffices.

Let $\gamma$ be a prime factor dividing $q^k-1$, it suffices to check
that $N_D \neq 0 \mod \gamma$, for~$\Psi$ to be well-defined in
subgroups of order a power of $\gamma$ in $\F_{q^{k}}^*$.  Since $N_D$
is defined as the least common multiple of the cardinalities of $\E$
over each of the finite fields $\F_{q^d}$ with $1\leq d \leq D$, where
$D\leq k$ is the maximum degree of the places we want to consider,
this gives us an extra condition on $\E$.  Namely, it should
satisfy the following property:
\begin{itemize}
\item For any $i=1, \cdots, D$, $ |\E/\F_{q^{i}}| \neq 0 \mod \gamma$.
\end{itemize}

Let us study this condition.  We denote by $t$ the trace of $\E$ over
$\F_q$
and  factor the characteristic polynomial of the Frobenius of $\E$:
\[X^2-tX+q = (X-r)(X-s) \mod \gamma
\] with $r$ and $s$ in $\F_{\gamma^2}$.  We know that the number of
points of $\E/\F_{q^i}$ is equal to $(1-r^{i})(1-s^{i}) \mod
\gamma$. Thus, to ensure that the cardinalities of $\E$
over the field extensions $\F_{q^i}$, with $i$ in $[1,D]$, all differ from $0$
modulo $\gamma$ we just need to verify that both $r^{i} \neq 1 \mod \gamma$ and
$s^{i} \neq 1 \mod \gamma$. In order to do that, let us first
study the order of the product $rs=q \mod \gamma$.

By definition of $\gamma$, we have $q^k=1\mod \gamma  $. Furthermore,
the order of $q$ is strictly smaller than $k$ modulo $\gamma$ if and
only if $\gamma$ already divides the order of the multiplicative group
of a subfield of $\F_{q^k }$. In that case, we compute this part of
the logarithm by applying our method to the smallest such subfield.

We now assume that the order of $q$ is precisely $k \mod \gamma$. Thus, for any $i$ not a multiple of $k$, at most one of
$r^i$ or $s^i$ can be equal to $1 \mod \gamma$. Exchanging $r$ and
$s$ if necessary, we now study the case $r^i=1 \mod \gamma.$
In that case, we have $s^i=q^i\neq 1\mod \gamma.$ This implies
that $\E(\F_{q^i})$ contains a $\gamma$-torsion point $Q_\gamma$ but
not the full $\gamma$-torsion $\E[\gamma]$. Thus, the Tate pairing
provides a non-degenerate bilinear map to the $\gamma$ roots of unity:
\[
  e_i: <Q_\gamma> \times~\E(\F_{q^i})/\gamma \E(\F_{q^i}) \mapsto \F_{q}^*/ \left(\F_{q}^*\right)^{(q^k-1)/\gamma}.
\]
Fixing an arbitrary non-zero element from $\E(\F_{q^i})/\gamma
\E(\F_{q^i})$, we obtain a linear map $\tilde{\Psi_i}$ from the subgroup
generated by $Q_\gamma$ to the $\gamma$-th roots of unity.

Possibly after renormalization, $\tilde{\Psi_i}$ gives a compatible
replacement for $\Psi$ that can be applied to the $\gamma$-torsion
point. The renormalization consists in replacing $\tilde{\Psi_i}$ by
$\tilde{\Psi_i}^{\beta_i}$, where $\beta_i$ is the renormalization
constant. As a consequence, it is mathematically possible to extend
$\Psi$ to all divisors. One computational caveat is that determining
the value of $\beta_i$ can be expressed as a discrete logarihm problem
in the group of order $\gamma$. It does not affect the efficiency of
the overall algorithm but prevents independent check on relations
containing divisors of degree $i$ not compatible with the definition of
$\Psi$ during the precomputation phase.

To see how $\beta_i$ can be determined, let us take a place $p_i$ of
degree $i$ and compute $\tilde{\Psi_i}(p_i)$ (as usual this it the
product of the value for all the conjugate points in $p_i$). Then apply
one step of the descent algorithm to relate $p_i$ to places of degrees
$\neq i$ which are all compatible with the computation of
$\Psi$. Multiplying these contributions gives the renormalized value
$\tilde{\Psi_i}(p_i)^{\beta_i}.$ Thus, if we wish to do so, we can
compute $\beta_i$ from the individual logarithms of these
two values.

\section{Relations for factor base extension}
\label{bunches}
Let us describe our decomposition in groups to extend the factor base to all elementary divisors
of height equal or lower than $4$. The idea is to write a partition of $q$ groups
with $q^2$ elements in each and to be able to decrease the height of the divisor associated to
the bracket on the right again.
To illustrate the process, we define a first group with the monomials:
\[
\begin{array}{rcl}
g_1&=&UV\\
 g_2&=&U+V\\
 g_3&=&1.
\end{array}
\]Defining then
$G$ as all the linear combinations of these three monomials with coefficients
in $\F_q$ permits to set our first (special) group as:
\[\mathcal{G}=\{\Xi(\Psi^*(g)) \,  | \, g \in G \}.
\] All the divisors in the special group have height lower than
$4$. We now sieve on pairs of polynomials $(A,B)$ such that
$A = g_1 + \alpha g_2$ and $B= g_1+\beta g_3$ where
$\alpha, \beta \in \F_q$. On the left side it is clear that all
polynomials raised in the product belong to $G$. So all the divisors
in the corresponding sum on the left side have height lower than $4$
(see Table~\ref{table}) and belong to $\mathcal{G}$. On the right
side, we are left with a bracket $\bracket A B$ leading to a height
lower than~$8$. Again, the probability that it splits into divisors
with a height lower than $3$ is too low. Yet, computing the brackets:
\[
\begin{array}{rcl}
\bracket {g_1} {g_2} &=&VW(U+V)-(V+W)UV=V^2(W-U) \\
\bracket {g_1} {g_3 } &=& V(W-U)\\
\bracket {g_2 }{g_3}  &=& W-U
\end{array}
\]
and thanks to bilinearity we obtain that $W-U$ is a common factor of $\bracket A B$.
Besides we note that $h(\Xi(\Phi^*(W-U)))=4$.

Removing this constant contribution, we are left with a residual
height of $4$ on the right side. If it decomposes into lower height
divisors, this gives us a linear equation involving the logarithms a subset of the
divisors in $\mathcal{G}$. With enough such equations, we again use
linear algebra to compute the logarithms of the elements of
$\mathcal{G}$.

Note that  the probability to find a good relation is $3/4$ when $q$
grows. We thus expect $3q^2/4$ equations in $q^2/4$.

Note that the pairs $(A,B)$ that fail to give an equation are
nonetheless useful! Indeed, a pair~$(A,B)$ of sieving polynomials
fails if $\bracket A B$  leads to a divisor with height
precisely~$4$. It means that after obtain the logarithm of elements of
$\mathcal{G}$  we can derive the logarithm of these extra divisors for
free.

\paragraph{Construction of groups with one compelled point.}
Following the idea of the special group $\mathcal{G}$, we would like
to construct small other groups of divisors that have two
properties. First, for each group, we need to be able to create
relations involving only heigh-$4$ divisors from this group on the
left, possibly with divisors of lower height. Second, we need to 
control the splitting probability of the bracket on the right of the
equation. We proceed using compelled points.

How to choose our generators $g_1, g_2$ and $g_3$ in this case ? We
recall that the naive height-$4$ sieving is based on the monomials
$1, U, V, UV$. Since there is no reason to favor nor $U$ neither $V$,
we propose to preserve symmetry between the two variables, writing:
\[
\begin{array}{rcl}
g_1&=&UV+k_1  U\\
 g_2&=&UV+k_2 V\\
 g_3&=&1.
\end{array}
\] where $k_1$ and $k_2$ are in the base field $\F_q$. Defining again groups:
\[\mathcal{G}_{k_1,k_2} = \{\Xi(\Phi^*(g_1+ \alpha g_2 + \beta g_3)) \,  | \, \alpha,\beta \in \F_q \} \]
with $q^2$ divisors each, we sieve on pairs of polynomials $(A,B)$ such that $A = g_1 + \alpha g_2$ and
$B= g_1+\beta g_3$ where $\alpha, \beta \in \F_q$. 
On the left side all divisors have height
lower than $4$ and belong to $\mathcal{G}_{k_1,k_2}$. On the right side, we are left
with a bracket  $\bracket A B$ leading to a height lower than $8$. To decrease this height we 
consider the brackets:
\[
\begin{array}{rcl}
\bracket {g_1} {g_2} &=& k_1 \bracket {U} {UV}  + k_2\bracket {UV} {V}  +k_1 k_2 \bracket {U} {V} \\
\bracket {g_1} {g_3 } &=& VW+k_1V- UV+k_1U\\
\bracket {g_2 }{g_3}  &=& VW+k_2W- UV+k_2V
\end{array}
\] Note that $\bracket A B$ is a linear combination of these brackets
and that the last two ones are associated to divisors of height lower
than $4$. Thus, removing a point in
$\Xi (\Phi^*(\bracket {g_1} {g_2}))$ will suffice. Let us look at
$k_1 \bracket {U} {UV} + k_2\bracket {UV} {V} +k_1 k_2 \bracket {U}
{V} $ in details. Calling $c_f$ the coefficient in $\F_q$ of the
leading monomial\footnote{Considering the weighted degree in $X$ and
  $Y$ of each monomial.}  of $\bracket {g_1} {g_2}$ for the
denominator of any fraction~$f$, we see that we can force
$k_1 c_{\bracket {U} {UV}} + k_2 c_{\bracket {UV} {V}} + k_1 k_2
c_{\bracket {U} {V}} = 0$ in $\F_q$. We underline that for any fixed
constant
$k_1 \neq - c_{\bracket {UV} {V}} c_{\bracket {U} {V}}^{-1} $ there
exists a unique $k_2$ such that the previous equality is verified. It
means that we create $q-1$ such groups.  Besides, this annihilates the
leading monomial so decreases the weighted degree of
$\bracket {g_1} {g_2}$ and leads to remove a point in the associated
divisor of $\bracket {A} {B}$.  Hence, we are left with a residual
height of $7$.  We want the corresponding divisor to be written as a
sum of divisors of height $3$ at most. The heuristic probability to
get a good relation is so equal to:
\[1- (1/7+1/6+1/5+1/4) \approx 0.2405
\] as $q$ tends to infinity.

This is slightly too low for the purpose. As a consequence, we need
either to improve the group construction or to make good use of the
equations with a single height-$4$ divisor in the
bracket. Nevertheless, since $0.24$ is close to $1/4$, it is
conceivable that we can find enough relations in practice. We decided
to test it and we computed all the discrete logarithms up to
height-$5$ for the target finite field $\F_{3^{1345}}=\F_{243^{269}}$
with this method.  

\begin{remark} It is useful to know that:
\[
\begin{array}{rcl}
\bracket {U} {UV} &=& UV(V-W) \\
\bracket {UV} {V } &=& VW(V-W)\\
\bracket {U}{V}  &=& V^2-UW
\end{array}
\]
\end{remark}

\paragraph{Interaction with the action of Frobenius.}
Looking at our groupings, we see that we have built a total of $q$
different ones (including the special group $\mathcal{G}$. Since each
grouping contains about $q^2/4$, the computations (if successful)
gives us about $q^3/4$ logarithms of height $4$. This is much less
than the total expected number which is close to $q^4/4$. However,
the action of the Frobenius potentially multiply these logarithms by a
factor of $k$. For practical, we heuristically assume that this is the
case. The fact that we were able to compute the logarithms of all height-$4$
divisors for $\F_{3^{1345}}$ supports this assumption.

\paragraph{Going to height $5$.}
We continue the extension to height $5$ in similar fashion. Since,
this requires more degree of freedom, we no longer need to use
compelled points. Instead, we sieve on more general polynomials of the
forms $A=UV+a_U U+a_1$ and $B=UV+b_V V+b_1$.
On the left-hand side, all factors of the form $A-\alpha B$ have
height $4$. Thus they decompose in divisors of height at most $4$ and
their logarithms can be directly obtained. On the right-hand side, the
bracket has height at most $8$. We expect that it contains an
elementary divisor of height $5$ with probability close to $1/5$. As a
consequence, we obtain about $q^4/5$ divisors of height $5$ without
performing any linear algebra.

Again, thanks to the action of Frobenius, we expect to recover an
overwhelming fraction of divisors of height $5$. This turn out to work
in practice for our example $\F_{3^{1345}}$.


\end{document}